\documentclass[a4paper,11pt]{article} 

\usepackage{fullpage}
\usepackage{times}
\usepackage{soul}
\usepackage{url}
\usepackage[utf8]{inputenc}
\usepackage[small]{caption}
\usepackage{graphicx}
\usepackage{xcolor}
\usepackage{amsmath}
\usepackage{booktabs}
\usepackage{algorithm}
\usepackage{enumerate}
\usepackage{dsfont}
\usepackage{cleveref}

\usepackage{amsthm}
\usepackage{amssymb}
\usepackage[noend]{algorithmic}

\usepackage{bbm}
\usepackage{multirow}

\usepackage{nicematrix}
\usepackage{thm-restate}

\newtheorem{theorem}{Theorem}[section]

\newtheorem{lemma}[theorem]{Lemma}

\theoremstyle{definition}

\newcommand*{\innerproofname}{Proof}

\usepackage{natbib}
\usepackage{authblk}

\renewcommand{\emptyset}{\varnothing}
\newcommand{\fbin}{{f_\textup{bin}}}
\newcommand{\fid}{{f_\textup{id}}}
\newcommand{\fidbin}{{f_\textup{id,bin}}}
\newcommand{\fidweak}{{\tilde{f}_\textup{id}}}
\newcommand{\ftrans}{{f_\textup{trans}}}
\DeclareMathOperator*{\argmax}{arg\,max}
\DeclareMathOperator*{\argmin}{arg\,min}

\allowdisplaybreaks

\title{\bf Reforming an Unfair Allocation by Exchanging Goods}

\author[1]{Sheung Man Yuen}
\author[2]{Ayumi Igarashi}
\author[3]{Naoyuki Kamiyama}
\author[1]{Warut Suksompong}

\affil[1]{National University of Singapore, Singapore}
\affil[2]{University of Tokyo, Japan}
\affil[3]{Kyoto University, Japan}

\date{\vspace{-10mm}}

\begin{document}

\maketitle

\begin{abstract}
Fairly allocating indivisible goods is a frequently occurring task in everyday life.
Given an initial allocation of the goods, we consider the problem of reforming it via a sequence of exchanges to attain fairness in the form of envy-freeness up to one good (EF1).
We present a vast array of results on the complexity of determining whether it is possible to reach an EF1 allocation from the initial allocation and, if so, the minimum number of exchanges required.
In particular, we uncover several distinctions based on the number of agents involved and their utility functions.
Furthermore, we derive essentially tight bounds on the worst-case number of exchanges needed to achieve EF1 when the initial allocation is balanced.
\end{abstract}

\section{Introduction}

Fair division is a research area that studies how to allocate scarce resources to interested agents in a fair manner.
Over the past several decades, a large body of work has developed on concepts and algorithms for fairly allocating various types of resources \citep{BramsTa96,Moulin03,Moulin19,RobertsonWe98}.
The developed theory has also been applied to many allocation scenarios in practice \citep{GoldmanPr14,BudishCaKe17,IgarashiYo23,HanSu24}.

A ubiquitous fair division problem is the allocation of \emph{indivisible goods}, such as books, furniture, paintings, and human resources.
While numerous fairness notions have been proposed for allocating indivisible goods, one of the most prominent notions is \emph{envy-freeness up to one good (EF1)}.
This notion requires that if an agent envies another agent, there must exist a good in the latter agent's bundle whose removal would make the envy disappear.
Besides its simplicity, a salient feature of EF1 is that an allocation satisfying it always exists and can be found, e.g., by the \emph{round-robin algorithm}, which lets the agents pick their favorite goods in a round-robin fashion. 
The allocation returned by this algorithm is also \emph{balanced}, meaning that the numbers of goods that any pair of agents receive differ by at most one.

The fair division literature typically assumes that there is a set of unallocated goods and the objective is to allocate them fairly.
In certain scenarios, however, an existing (possibly unfair) allocation of the goods is already in place, and the goal is to ``reform'' it in order to arrive at a fair allocation.
This is the case, for instance, in the allocation of personnel among teams in an organization.
As the personnel experience individual growth or decline, and as the needs of the teams evolve, these changes can necessitate a reevaluation and potential reformation of the current allocation by the organization leadership.
Another example is the distribution of a museum's exhibits across its branches---the museum director may decide to adjust the distribution so as to ensure fairness based on the most recent valuations.
Such scenarios fall under the umbrella of \emph{control problems}, which assume a central authority with the power to control the outcome and have been studied extensively in fair division and computational social choice \citep{FaliszewskiRo16,ChenKaNu25}.

In this work, we shall allow agents to \emph{exchange} pairs of goods in the reformation process, and use EF1 as our fairness criterion.
Exchanges constitute a fundamental type of operation and preserve the size of each agent's bundle, thereby ensuring that any cardinality constraints remain fulfilled.\footnote{Capacity constraints are prevalent in fair division applications and have accordingly received interest in the literature \citep{BiswasBa18,HummelHe22,ShoshanSeHa23,WuLiGa25}.}
Naturally, given an initial allocation, we wish to reach an EF1 allocation using a small number of exchanges.
However, it is sometimes impossible to reach an EF1 allocation via \emph{any} finite number of exchanges (e.g., if one agent receives many more goods than another agent in the initial allocation), so we start by exploring whether the corresponding decision problem can be answered efficiently.
Since this problem is equivalent to determining whether an EF1 allocation with a certain size vector exists in a given instance, it is also meaningful independently of exchange considerations.\footnote{A \emph{size vector} specifies the number of goods that each agent receives in an allocation.
When an EF1 allocation is not guaranteed to exist in some instances due to cardinality requirements, an approach taken by previous work is to relax the EF1 condition (e.g., \citep{WuLiGa25}).
However, this leads to unnecessarily weak guarantees in instances where EF1 can be attained.}
We also investigate other important questions in this setting.
Namely, if it is possible to reach an EF1 allocation, can we efficiently determine the smallest number of exchanges required to achieve this goal?
And how many exchanges might we need to make in the worst case in order to attain EF1?

\subsection{Our Results}

In our model, there is an initial allocation of a set of goods.
As is often assumed in fair division, each agent has an additive utility function over the goods.
At each step of the reformation process, two agents can exchange a pair of goods with each other to obtain another allocation, and the goal is to reach an EF1 allocation at the end of the process.
More details on our model are provided in \Cref{sec:prelim}.

In \Cref{sec:exist}, we consider the decision problem of determining whether a given initial allocation can be reformed into an EF1 allocation.
As mentioned earlier, this problem is equivalent to determining whether an EF1 allocation with a given size vector exists, so we focus on the latter decision problem instead.
We demonstrate interesting distinctions in the complexity based on the number of agents and their utility functions.
Specifically, in the case of two agents, the problem can be solved in polynomial time if the agents have identical utilities, but becomes (weakly) NP-complete otherwise.
For three or more agents, the problem is NP-complete even with identical utilities; however, it can be solved efficiently when the utilities are binary\footnote{This means that each agent has utility $0$ or $1$ for each good.} provided that the number of agents is constant.
Finally, for an arbitrary (non-constant) number of agents, the problem is strongly NP-hard even for identical utilities and NP-hard for binary utilities,\footnote{For binary utilities, there is no distinction between weak and strong NP-hardness.} but can be solved in polynomial time if the utilities are identical \emph{and} binary.
The results of this section are summarized in \Cref{tab:exists}.

\begin{table*}[tb]
\centering
\begin{NiceTabular}{c||c|c|c|c}
utilities
    & general
    & identical
    & binary
    & identical \& binary \\ \hline \hline
$n = 2$
    & wNP-c (Thm.~\ref{thm:exist_gen_two_nphard})
    & P (Thm.~\ref{thm:exist_iden_two})
    & P (Thm.~\ref{thm:exist_bin_const})
    & P (Thm.~\ref{thm:exist_bin_const}) \\ \hline
constant $n \geq 3$
    & wNP-c (Thm.~\ref{thm:exist_iden_const_nphard})
    & wNP-c (Thm.~\ref{thm:exist_iden_const_nphard})
    & P (Thm.~\ref{thm:exist_bin_const})
    & P (Thm.~\ref{thm:exist_bin_const}) \\ \hline
general $n$
    & sNP-c (Thm.~\ref{thm:exist_iden_gen_nphard})
    & sNP-c (Thm.~\ref{thm:exist_iden_gen_nphard})
    & NP-c (Thm.~\ref{thm:exist_bin_gen_nphard})
    & P (Thm.~\ref{thm:exist_idenbin_gen})
\end{NiceTabular}
\caption{Computational complexity of \textsc{Reformability}, the problem of deciding whether an EF1 allocation with a given size vector exists in a given instance. The top row represents the class of utility functions considered. The leftmost column represents the number of agents. ``sNP-c'' and ``wNP-c'' stand for strongly NP-complete and weakly NP-complete respectively.} \label{tab:exists}
\end{table*}

Having determined the reformability of the initial allocation, we next explore in \Cref{sec:optimal} the problem of deciding whether the optimal (i.e., minimum) number of exchanges required to reach an EF1 allocation is at most some given number $k$.
For (a) two agents with identical utilities, (b) a constant number of agents with binary utilities, and (c) any number of agents with identical binary utilities, we show that this problem can be solved in polynomial time, just like the corresponding decision problem in \Cref{sec:exist}.
For the remaining scenarios, since deciding whether an allocation is reformable is already NP-complete (from \Cref{sec:exist}), we instead focus on the special case where the allocation is balanced---an EF1 allocation is guaranteed to be reachable in this case (see \Cref{prop:exist_balanced}).
We show that the problem for this special case remains NP-complete.

In \Cref{sec:worst}, instead of considering specific instances, we derive worst-case bounds on the number of exchanges necessary to reach an EF1 allocation.
We assume that each of the $n$ agents possesses $s$ goods---this again ensures that an EF1 allocation is reachable.
We show that roughly $s(n-1)/2$ exchanges always suffice for instances with general utilities or with binary utilities; moreover, our bound is essentially tight for any $n$ and $s$, and exactly tight when $n = 2$ as well as when $s$ is divisible by $n$.
For instances with identical binary utilities, we show that an essentially tight bound for the number of exchanges is $sn/4$ for even $n$ and $s(n-1)(n+1)/4n$ for odd $n$.

While our focus in this paper is on exchanges, we show in Appendix~\ref{ap:transfer} that some of our results hold analogously when agents are allowed to \emph{transfer} goods rather than exchanging them.

\subsection{Additional Related Work} \label{sec:related-work}

As mentioned earlier, the majority of work in fair division assumes that there is no initial allocation of the resources---we now discuss the key exceptions and their differences from our model.
\citet{BoehmerBrHe24} studied the problem of \emph{discarding} goods from an initial allocation in order to reach an envy-free or EF1 allocation.
As it is possible to reallocate the goods in several practical situations, discarding them can be unnecessarily wasteful for the agents involved.
In a similar vein, \citet{DornDeSc21} investigated deleting goods to attain another fairness notion called \emph{proportionality}; they assumed that agents have ordinal preferences (rather than cardinal utilities) over the goods, and considered both the settings with and without an initial allocation.\footnote{When there is no initial allocation, \citet{DornDeSc21}~considered deleting goods so that a proportional allocation of the remaining goods exists. 
\citet{AzizScWa16} examined discarding or adding goods to achieve envy-freeness, also in the absence of an initial allocation and under ordinal preferences.}
Instead of deleting goods, \citet{BredereckKaLu23} allowed agents to share goods in order to improve their allocations, while \citet{BentertBrDe25} analyzed the resolution of envy by adding goods.
\citet{AzizBiLa19} focused on reallocating goods to make agents better off, but did not delve into the aspect of fairness.
\citet{IgarashiKaSu24} aimed to transition from an initial allocation to a target allocation, both of which are EF1, while maintaining EF1 throughout the process.
\citet{ChandramouleeswaranNiRa25} examined \emph{transferring} goods starting from a ``near-EF1'' allocation with the goal of reaching an EF1 allocation.
\citet{Segalhalevi22} considered the reallocation of a \emph{divisible} good and explored the trade-off between guaranteeing a minimum utility for every agent and ensuring each agent a certain fraction of her original utility.
\citet{ChevaleyreEnEs07} also strived to reach fair allocations but via exchanges with money.

Further afield, the idea of improving an initial allocation has also been examined when each agent receives only one good, a setting sometimes known as a \emph{housing market}.
\citet{GourvesLeWi17} assumed an underlying social network and allowed beneficial exchanges between agents who are neighbors in the network---their work led to a series of follow-up papers on similar models \citep{HuangXi20,LiPlSi21,MullerBe21,ItoKaKa23}.
\citet{DamammeBeCh15} also considered exchanges but without an underlying graph structure, while \citet{BrandtWi24} used \emph{Pareto-optimality} as their target notion.
The papers mentioned so far in this paragraph did not have fairness as their objective.
\citet{ItoIwKa25} incorporated fairness in the form of envy-freeness into this setting---starting with an envy-free allocation, they let each agent exchange her current good with a preferred unassigned good as long as the exchange keeps the allocation envy-free.

Finally, our work is similar in spirit to that of \citet{BredereckChWo16} on voting, which addressed the following question: given a voting instance, is there a ``nicely structured'' instance nearby?
In particular, we use an EF1 allocation as a target and measure distance in terms of the number of exchanges.

\section{Preliminaries} \label{sec:prelim}

Let $N = \{1, \ldots, n\}$ be a set of $n \geq 2$ agents, and $M$ be a set of $m \geq 1$ goods typically denoted by $g_1, \ldots, g_m$. 
A \emph{bundle} is a (possibly empty) subset of goods.
An \emph{allocation} $\mathcal{A} = (A_1, \ldots, A_n)$ is an ordered partition of $M$ into $n$ bundles such that bundle $A_i$ is allocated to agent $i \in N$.
An \emph{(allocation) size vector} $\vec{s} = (s_1, \ldots, s_n)$ is a vector of non-negative integers such that $\sum_{i\in N} s_i = m$.
We say that an allocation $\mathcal{A}$ has size vector $\vec{s}$ if $|A_i| = s_i$ for all $i \in N$.
A size vector $\vec{s}$ is \emph{balanced} if $|s_i - s_j| \leq 1$ for all $i, j \in N$, and an allocation is \emph{balanced} if it has a balanced size vector.

Each agent $i \in N$ has an additive \emph{utility function} $u_i: 2^M \to \mathbb{R}_{\geq 0}$ that maps bundles to non-negative real numbers; additivity means that $u_i (M') = \sum_{g \in M'} u_i (\{g\})$ for all $i \in N$ and $M' \subseteq M$.
We write $u_i (g)$ instead of $u_i (\{g\})$ for a single good $g \in M$. 
The utility functions are \emph{identical} if $u_i = u_j$ for all $i, j \in N$---we shall use $u$ to denote the common utility function in this case.
The utility functions are \emph{binary} if $u_i(g) \in \{0, 1\}$ for all $i \in N$ and $g \in M$.
Agent~$i$ is \emph{EF1 towards} agent $j$ in an allocation $\mathcal{A} = (A_1, \ldots, A_n)$ if either $A_j = \emptyset$ or there exists a good $g \in A_j$ such that $u_i(A_i) \geq u_i(A_j \setminus \{g\})$.
An allocation $\mathcal{A}$ is \emph{EF1} if every agent is EF1 towards every other agent in $\mathcal{A}$.
A \emph{(fair division) instance} $\mathcal{I}$ consists of a set of agents~$N$, a set of goods $M$, and the agents' utility functions $(u_i)_{i \in N}$.

An \emph{exchange} involves an agent $i$ giving one good from her bundle to agent $j$ and simultaneously receiving one good from agent $j$'s bundle.
We say that allocations $\mathcal{A} = (A_1, \ldots, A_n)$ and $\mathcal{B} = (B_1, \ldots, B_n)$ can be \emph{reached via an exchange} if there exist distinct $i, j \in N$, $g \in A_i$, and $g' \in A_j$ such that $B_i = (A_i \cup \{g'\}) \setminus \{g\}$, $B_j = (A_j \cup \{g\}) \setminus \{g'\}$, and $B_k = A_k$ for all $k \in N \setminus \{i, j\}$.
An allocation $\mathcal{B}$ can be \emph{reached} from an allocation $\mathcal{A}$ if there exist a non-negative integer $T$ and a sequence of allocations $(\mathcal{A}^0, \mathcal{A}^1, \ldots, \mathcal{A}^T)$ such that $\mathcal{A}^0 = \mathcal{A}$, $\mathcal{A}^T = \mathcal{B}$, and for each $t \in \{0, \ldots, T-1\}$, $\mathcal{A}^t$ and~$\mathcal{A}^{t+1}$ can be reached via an exchange.
The \emph{optimal number of exchanges} required to reach $\mathcal{B}$ from $\mathcal{A}$ is the smallest~$T$ across all such sequences of allocations---if no such $T$ exists (i.e., $\mathcal{B}$ cannot be reached from~$\mathcal{A}$), then the optimal number of exchanges is defined\footnote{The optimal number of exchanges can be viewed as the distance between the two allocations in the implicit \emph{exchange graph}, where the allocations are vertices and the edges connect allocations that can be reached via an exchange \citep{IgarashiKaSu24}.
When there is no path connecting two vertices of a graph, it is common to define the distance between them as infinity.} to be $\infty$.
Observe that two allocations can be reached from each other if and only if they share the same size vector.

\begin{restatable}{proposition}{propreachable} \label{prop:reachable}
Given an instance, let $\mathcal{A}$ and $\mathcal{B}$ be allocations in the instance. Then, $\mathcal{B}$ can be reached from $\mathcal{A}$ if and only if $\mathcal{A}$ and $\mathcal{B}$ have the same size vector.
\end{restatable}

\begin{proof}
Note that every exchange preserves the size vector of the allocation, since each agent involved in the exchange gives away one good and receives one good in return, while other agents retain their bundles.

$(\Rightarrow)$ If $\mathcal{B}$ can be reached from $\mathcal{A}$, then there exist a non-negative integer $T$ and a sequence of allocations $(\mathcal{A}^0, \mathcal{A}^1, \ldots, \mathcal{A}^T)$ such that $\mathcal{A}^0 = \mathcal{A}$, $\mathcal{A}^T = \mathcal{B}$, and for each $t \in \{0, \ldots, T-1\}$, $\mathcal{A}^t$ and~$\mathcal{A}^{t+1}$ can be reached via an exchange. 
For each $t \in \{0, \ldots, T-1\}$, $\mathcal{A}^t$ and~$\mathcal{A}^{t+1}$ have the same size vector.
Therefore, the whole sequence of allocations, including $\mathcal{A}$ and $\mathcal{B}$, have the same size vector.

$(\Leftarrow)$ Assume that $\mathcal{A}$ and $\mathcal{B}$ have the same size vector; we shall create a sequence of allocations from $\mathcal{A}$ to $\mathcal{B}$ with the desired properties.
We first remedy the goods in agent $1$'s bundle.
If $A_1 = B_1$, then all the goods in agent $1$'s bundle are correct and we are done.
Otherwise, since $|A_1| = |B_1|$, we must have $|A_1 \setminus B_1| = |B_1 \setminus A_1| > 0$.
Perform an exchange between a good $g \in A_1 \setminus B_1$ and a good $g' \in B_1 \setminus A_1$.
This creates a new allocation where the number of wrong goods in agent $1$'s bundle decreases by one.
By repeating this procedure, we eventually arrive at an allocation with agent $1$'s bundle remedied.
We then remedy the goods in the bundles of agents $2, 3, \ldots, n$ in the same manner until every agent has her own bundle in $\mathcal{B}$.
Note that when the goods in agent $i$'s bundle are remedied, there is no exchange of goods involving agents $1$ to $i-1$ anymore, and so the bundles of agents $1$ to $i-1$ remain correct.
This shows that $\mathcal{B}$ can be reached from~$\mathcal{A}$.
\end{proof}

Next, we state a simple proposition that characterizes the existence of EF1 allocations based on the size vector.

\begin{restatable}{proposition}{propexistbalanced} \label{prop:exist_balanced}
Let $\vec{s} = (s_1, \ldots, s_n)$, and let $m = \sum_{i=1}^n s_i$.
\begin{enumerate}
    \item[(a)] If $\vec{s}$ is balanced, then every instance with $n$ agents and $m$ goods admits an EF1 allocation with size vector $\vec{s}$.
    \item[(b)] If $\vec{s}$ is not balanced, then there exists an instance with $n$ agents and $m$ goods that does not admit any EF1 allocation with size vector $\vec{s}$.
\end{enumerate}
\end{restatable}

\begin{proof}
\begin{enumerate}
    \item[(a)] An EF1 allocation can be guaranteed by allowing agents to pick their favorite goods in a round-robin fashion, with agents with higher $s_i$ (if any) starting before those with lower $s_i$, until each agent $i$ has $s_i$ goods.
    \item[(b)] Let $\mathcal{I}$ be an instance with $m$ goods such that $u_i(g) = 1$ for all $i \in N$ and $g \in M$.
    Let $\mathcal{A}$ be any allocation with size vector $\vec{s}$.
    Since $\vec{s}$ is not balanced, there exist distinct $i, j \in N$ such that $s_j - s_i \geq 2$.
    Then, we have $|A_j| = s_j > 0$, so $A_j \neq \emptyset$.
    Furthermore, $u_i(A_i) = s_i < s_j - 1 = u_i(A_j \setminus \{g\})$ for all $g \in A_j$.
    This shows that agent $i$ is not EF1 towards agent $j$, and so $\mathcal{A}$ is not EF1.
    Therefore, $\mathcal{I}$ does not admit an EF1 allocation with size vector~$\vec{s}$. \qedhere
\end{enumerate}
\end{proof}

Finally, we introduce an NP-hard decision problem called the \textsc{Balanced Multi-Partition} problem, which we will use later in the proofs of several results.
In \textsc{Balanced Multi-Partition}, we are given positive integers $p, q, K$ and a multiset of positive integers $X = \{x_1, \ldots, x_{pq}\}$ such that $K < x_j \leq 2K$ for all $j \in \{1, \ldots, pq\}$, and the sum of all the integers in $X$ is $p(q+1)K$.
The problem is to decide whether $X$ can be partitioned into multisets $X_1, \ldots, X_p$ of equal cardinalities and sums, i.e., for each $i \in \{1, \ldots, p\}$, the cardinality of $X_i$ is $q$ and the sum of all the integers in $X_i$ is $(q+1)K$.
The NP-hardness of this problem is based on a reduction from the equal-cardinality version of the NP-hard problem \textsc{Partition} \citep[p.~223]{GareyJo79}.

\begin{restatable}{proposition}{propbalancedmultipartition} \label{prop:balanced_multi_partition}
For any fixed $p \geq 2$, \textsc{Balanced Multi-Partition} is NP-hard.
\end{restatable}

\begin{proof}
We shall prove NP-hardness via a series of reductions from the equal-cardinality version of \textsc{Partition}.
In this version, we are given positive integers $q, K'$ and a multiset of positive integers $W = \{w_1, \ldots, w_{2q}\}$ such that the sum of the integers in $W$ is $2K'$.
The problem is to decide whether $W$ can be partitioned into multisets $W_1$ and $W_2$ of equal cardinalities and sums.
This problem is known to be NP-hard \citep[p.~223]{GareyJo79}.

Let an instance of the equal-cardinality version of \textsc{Partition} be given, and let $p \geq 2$ be a fixed integer.
If some integer in $W$ is more than $K'$, then $W$ cannot be partitioned into the desired multisets; therefore, we assume that every integer in $W$ is at most $K'$.
Define a multiset $W^1 = \{w_j \mid j \in \{2q+1, \ldots, 2q+(p-2)\}\}$ such that $w_j = K'$ for all $w_j \in W^1$; define a multiset $W^0 = \{w_j \mid j \in \{2q+(p-2)+1, \ldots, pq\}\}$ such that $w_j = 0$ for all $w_j \in W^0$; and define $W' = W \cup W^1 \cup W^0$.
Essentially, we are adding $p-2$ copies of the number $K'$ and sufficiently many copies of the number $0$ so that the total number of elements in $W'$ is $pq$.
Note that every integer in $W'$ is at most $K'$, and the sum of all the elements in $W'$ is $pK'$.
We claim that $W'$ can be partitioned into multisets $W'_1, \ldots, W'_p$ of equal cardinalities and sums (i.e., each $W'_i$ has cardinality $q$ and sum $K'$) if and only if $W$ can be partitioned into multisets $W_1$ and $W_2$ of equal cardinalities and sums.

$(\Leftarrow)$ If we are given a partition into multisets $W_1$ and $W_2$, let $W'_1 = W_1$, let $W'_2 = W_2$, and let $W'_i$ contain one element from $W^1$ and $q-1$ elements from $W^0$ for each $i \in \{3, \ldots, p\}$.
Each of $W'_1, \ldots, W'_p$ has $q$ elements with sum $K'$.
This gives a desired partition of $W'$.

$(\Rightarrow)$ Assume that we are given a partition into multisets $W'_1, \ldots, W'_p$.
If some $W'_i$ contains at least two elements in $W^1$, then the sum of $W'_i$ is more than $K'$, which is not possible.
Therefore, every $W'_i$ contains at most one element in $W^1$.
Furthermore, for each $W'_i$ containing some element in $W^1$, if it contains some element in $W$, then its sum would exceed $K'$, which is again not possible.
Therefore, there are $p-2$ of the $W'_i$ such that each of them contains one element from $W^1$ and $q-1$ elements from $W^0$.
This means that two of the $W'_i$ contain exactly the elements in $W$.
These two $W'_i$ induce the desired partition into $W_1$ and $W_2$ of $W$.

Now, define an instance of \textsc{Balanced Multi-Partition} as follows.
Let $K = K'+q$, and let $X = \{x_1, \ldots, x_{pq}\}$ be such that $x_j = w_j + K + 1$ for all $j \in \{1, \ldots, pq\}$.
For each $j \in \{1, \ldots, pq\}$, since $0 \leq w_j \leq K'$, we have $K < x_j \leq K' + K + 1 \leq 2K$.
The sum of all integers in $X$ is $pK' + pq(K+1) = p(q+1)K$.
It is clear that $X$ can be partitioned into multisets $X_1, \ldots, X_p$ of equal cardinalities and sums if and only if $W'$ can be partitioned into multisets $W'_1, \ldots, W'_p$ of equal cardinalities and sums, since the difference between $x_j$ and $w_j$ is the same for all $j$.
Note that the reductions in this proof can all be done in polynomial time.
This establishes the NP-hardness of \textsc{Balanced Multi-Partition}.
\end{proof}

\section{Reformability of Allocations} \label{sec:exist}

We start by investigating the decision problem of whether a given initial allocation can be reformed into an EF1 allocation.
By \Cref{prop:reachable}, this reformation is possible if and only if there exists an EF1 allocation with the same size vector as the initial allocation.
Therefore, in the rest of this section, we shall equivalently focus on the problem of deciding the existence of an EF1 allocation with a given size vector---this problem can be of interest independently of reformation considerations, e.g., when space constraints are present.

Now, \Cref{prop:exist_balanced} tells us that an EF1 allocation with a balanced size vector always exists.
This means that the only time when we may have difficulties in ascertaining whether an EF1 allocation exists is when the given size vector is \emph{not} balanced.
In fact, as some of our proofs in this section show, the decision problem is NP-hard even when the sizes of the agents' bundles differ by \emph{exactly two} (e.g., in \Cref{thm:exist_gen_two_nphard}).

We discuss the cases of two agents, a constant number of agents, and a general number of agents separately.
For each case, we explore how the hardness of the decision problem varies across different classes of utility functions.
Our results are summarized in \Cref{tab:exists}.

For convenience, we refer to as \textsc{Reformability} the problem of deciding whether an EF1 allocation with a given size vector exists in a given instance.
Note that \textsc{Reformability} is in NP regardless of the number of agents, as we can verify in polynomial time whether a given allocation satisfies the condition by simply checking its size vector and comparing the bundles of every pair of agents for EF1.

\subsection{Two Agents}

For two agents, interestingly, the computational complexity of the problem turns out to be different depending on whether the agents have identical utilities or not.
We begin our discussion with the case of identical utilities.

For two agents with identical utilities, we first provide a characterization for the existence of a desired EF1 allocation based on the size vector and the utilities of the goods. 
We show that an EF1 allocation with a given size vector exists if and only if the agent with fewer goods (say, agent $1$) is EF1 towards the other agent (say, agent $2$) in the allocation where agent $1$ receives the most valuable goods.
Note that the condition in the lemma only requires checking that agent~$1$ is EF1 towards agent $2$; in particular, it does not require checking that agent~$2$ is EF1 towards agent $1$.

\begin{restatable}{lemma}{lemexistidentwo} \label{lem:exist_iden_two}
Given an instance with two agents with identical utilities, let $\vec{s} = (s_1, s_2)$ be a size vector with $s_1 \leq s_2$. 
Assume that the goods $g_1, \ldots, g_m$ are arranged in non-increasing order of utility, and let $M_0 = \{g_1, \ldots, g_{s_1}\}$. 
Then, there exists an EF1 allocation with size vector $\vec{s}$ if and only if agent~$1$ is EF1 towards agent~$2$ in the allocation $(M_0, M \setminus M_0)$.
\end{restatable}

\begin{proof}
We say in this proof that for any nonempty set $M' \subseteq M$, the good $g_i \in M'$ is the \emph{most valuable} good in $M'$ if $g_i$ is the good with the smallest index in $M'$; likewise, $g_i$ is the \emph{least valuable} good in $M'$ if $g_i$ is the good with the largest index in $M'$.
Note that the most (resp.~least) valuable good in $M'$ is the one with the highest (resp.~lowest) utility among all the goods in $M'$, with ties broken by index.

$(\Rightarrow)$ Let $(A_1, A_2)$ be an EF1 allocation with size vector~$\vec{s}$.
Let $g$ and $g'$ be the most valuable good in $A_2$ and $M \setminus M_0$ respectively.
Since $M_0$ is the set containing the $s_1$ most valuable goods, we have $u(M_0) \geq u(A_1)$.
Since $(A_1, A_2)$ is an EF1 allocation, we have $u(A_1) \geq u(A_2 \setminus \{g\})$.
Moreover, since $M \setminus M_0$ is the set containing the $s_2$ least valuable goods, we have $u(A_2 \setminus \{g\}) \geq u((M \setminus M_0) \setminus \{g'\})$.
Combining the three inequalities, we get $u(M_0) \geq u((M \setminus M_0) \setminus \{g'\})$.
It follows that agent~$1$ is EF1 towards agent~$2$ in the allocation $(M_0, M \setminus M_0)$.

$(\Leftarrow)$ Suppose that agent $1$ is EF1 towards agent $2$ in the allocation $(M_0, M \setminus M_0)$.
If agent $2$ is also EF1 towards agent $1$ in $(M_0, M \setminus M_0)$, then we are done; therefore, assume that agent $2$ envies agent $1$ by more than one good.
For notational simplicity, let $h_j = g_{s_1 + j}$ for $j \in \{1, \ldots, s_1\}$, so that the goods arranged in non-increasing order of utility are $g_1, g_2, \ldots, g_{s_1}, h_1, h_2, \ldots, h_{s_1}, g_{2s_1+1}, \ldots, g_m$.
Let $A_1^1 = M_0 = \{g_1, \ldots, g_{s_1}\}$ and $A_2^1 = M \setminus M_0 = \{h_1, \ldots, h_{s_1}\} \cup \{g_{2s_1+1}, \ldots, g_m\}$.

Let $t = 1$. In the allocation $(A_1^t, A_2^t)$, agent $1$ is EF1 towards agent $2$, but agent $2$ envies agent $1$ by more than one good. Since $g_t$ is the most valuable good in $A_1^t$, we have $u(A_2^t) < u(A_1^t \setminus \{g_t\})$. Let $A_1^{t+1} = (A_1^t \cup \{h_t\}) \setminus \{g_t\}$ and $A_2^{t+1} = (A_2^t \cup \{g_t\}) \setminus \{h_t\}$ be the bundles after exchanging $g_t$ and $h_t$. Then, we have
\begin{align*}
    u(A_1^{t+1}) &= u((A_1^t \cup \{h_t\}) \setminus \{g_t\}) \\
    &\geq u(A_1^t \setminus \{g_t\}) 
    > u(A_2^t) 
    = u((A_2^{t+1} \cup \{h_t\}) \setminus \{g_t\}) 
    \geq u(A_2^{t+1} \setminus \{g_t\}),
\end{align*}
so agent $1$ is EF1 towards agent $2$ in $(A_1^{t+1}, A_2^{t+1})$. If agent~$2$ is also EF1 towards agent $1$, then $(A_1^{t+1}, A_2^{t+1})$ is an EF1 allocation and we are done.
Otherwise, agent $2$ envies agent~$1$ by more than one good, and we increment $t$ by $1$ and repeat the discussion in this paragraph.

If we still have not found an EF1 allocation after $t = s_1$, then agent $1$ is EF1 towards agent $2$ in $(A_1^{s_1+1}, A_2^{s_1+1})$, where $A_1^{s_1+1} = \{h_1, \ldots, h_{s_1}\} \subseteq A_2^1$ and $A_2^{s_1+1} = \{g_1, \ldots, g_{s_1}\} \cup \{g_{2s_1+1}, \ldots, g_m\} \supseteq A_1^1$, and $g_1$ is the most valuable good in $A_2^{s_1+1}$. 
This implies that
\begin{align*}
    u(A_1^{s_1+1}) &\leq u(A_2^1) 
    < u(A_1^1 \setminus \{g_1\}) 
    \leq u(A_2^{s_1+1} \setminus \{g_1\}),
\end{align*}
which means that agent $1$ is \emph{not} EF1 towards agent $2$ in $(A_1^{s_1+1}, A_2^{s_1+1})$. This is a contradiction; therefore, $(A_1^t, A_2^t)$ must be EF1 for some $t\in\{1,\dots,s_1\}$. 
\end{proof}

Since the condition in \Cref{lem:exist_iden_two} can be checked in polynomial time, we can derive the following result.

\begin{restatable}{theorem}{thmexistidentwo} \label{thm:exist_iden_two}
\textsc{Reformability} is in P for two agents with identical utilities.
\end{restatable}

\begin{proof}
Without loss of generality, let the size vector be $(s_1, s_2)$ with $s_1 \leq s_2$.
Arrange the goods $g_1, \ldots, g_m$ in non-increasing order of utility (which can be done in polynomial time), and let $M_0$ be the set of $s_1$ goods with the highest utilities.
By \Cref{lem:exist_iden_two}, there exists an EF1 allocation with size vector $(s_1, s_2)$ if and only if agent $1$ is EF1 towards agent $2$ in the allocation $(M_0, M \setminus M_0)$.
The latter condition can be checked in polynomial time.
\end{proof}

While deciding whether an EF1 allocation with a given size vector exists can be done efficiently for two agents with identical utilities, we remark here that deciding whether an \emph{envy-free} allocation exists is NP-hard for two agents with identical utilities even if we allow any size vector---this follows directly from a reduction from \textsc{Partition}.\footnote{If we require both agents to receive the same number of goods, the problem for envy-freeness remains NP-hard by a reduction from the equal-cardinality version of \textsc{Partition}.}

We now proceed to general utilities.
\Cref{lem:exist_iden_two} assumes identical utilities, and there is no obvious way to generalize it to non-identical utilities.
In fact, perhaps surprisingly, we show that the decision problem becomes NP-hard when we drop the assumption of identical utilities.
The proof follows from a reduction from \textsc{Balanced Multi-Partition} with $p = 2$, an NP-hard problem by \Cref{prop:balanced_multi_partition}.

\begin{restatable}{theorem}{thmexistgentwonphard} \label{thm:exist_gen_two_nphard}
\textsc{Reformability} is weakly NP-complete for two agents.
\end{restatable}

\begin{proof}
Clearly, this problem is in NP.
The ``weak'' aspect is demonstrated later in \Cref{lem:exist_gen_const_pseudopoly}, which says that there exists a pseudopolynomial-time algorithm that solves this problem for any constant number of agents.
Therefore, it suffices to show that this problem is NP-hard.

To demonstrate NP-hardness, we shall reduce from the NP-hard problem \textsc{Balanced Multi-Partition} with $p = 2$ (see \Cref{prop:balanced_multi_partition}).
Let a \textsc{Balanced Multi-Partition} instance be given with $p = 2$.
Without loss of generality, assume that $q \geq 2$.
Let $Y = \{y_1, \ldots, y_{2q+2}\}$ be a multiset such that $y_j = x_j$ for $j \in \{1, \ldots, 2q\}$, $y_{2q+1} = 2K$, and $y_{2q+2} = 0$.
We claim that $Y$ can be partitioned into two multisets $Y_1$ and $Y_2$ of equal cardinalities (i.e., of size $q+1$ each) with sums $(q+3)K$ and $(q+1)K$ respectively if and only if $X$ can be partitioned into two multisets $X_1$ and $X_2$ of equal cardinalities and sums.
If the latter condition holds, then let $Y_1$ (resp.~$Y_2$) contain the corresponding elements in $X_1$ (resp.~$X_2$), and let $y_{2q+1} \in Y_1$ and $y_{2q+2} \in Y_2$---this gives an appropriate partition of~$Y$.
Conversely, if the former condition holds, then we show that $X$ can be partitioned appropriately.
Note that if $y_{2q+1} \in Y_2$, then there are at least $q-1 > 0$ integers in $\{y_1, \ldots, y_{2q}\}$ that are also in $Y_2$. 
Since every integer in $\{y_1, \ldots, y_{2q}\}$ is more than $K$, the sum of $Y_2$ will be more than $(q-1)K + 2K = (q+1)K$, which is a contradiction.
This means that $y_{2q+1} \in Y_1$.
Similarly, if $y_{2q+2} \in Y_1$, then there are exactly $q+1$ integers in $\{y_1, \ldots, y_{2q}\}$ that are in $Y_2$.
The sum of $Y_2$ will be more than $(q+1)K$, which is a contradiction.
Hence, $y_{2q+2} \in Y_2$.
Now, this means that $\{y_1, \ldots, y_{2q}\}$ must be partitioned into two multisets of equal cardinalities (i.e., of size $q$ each) 
with sum $(q+1)K$ each.
This induces an appropriate partition of~$X$.

Next, define a fair division instance as follows.
There are $n = 2$ agents and a set of goods $M = \{g_1, \ldots, g_{2q+6}\}$.
Agent $2$'s utility is such that $u_2(g_j) = y_j$ for $j \in \{1, \ldots, 2q+2\}$, $u_2(g_{2q+3}) = u_2(g_{2q+4}) = 0$, and $u_2(g_{2q+5}) = u_2(g_{2q+6}) = 2K$.
Agent $1$'s utility is such that $u_1(g) = u_2(g) + 4K$ for all $g \in M$.
The size vector $\vec{s}$ is $(q+2, q+4)$.
This reduction can be done in polynomial time.
We claim that there exists an EF1 allocation with size vector $\vec{s}$ in this instance if and only if $Y$ can be partitioned into two multisets $Y_1$ and $Y_2$ of equal cardinalities (i.e., of size $q+1$ each) with sums $(q+3)K$ and $(q+1)K$ respectively.

$(\Leftarrow)$ Let $J'_1$ and $J'_2$ be a partition of $\{1, \ldots, 2q+2\}$ of equal cardinalities such that $\sum_{j \in J'_1} y_j = (q+3)K$ and $\sum_{j \in J'_2} y_j = (q+1)K$.
Let $A_1 = \{g_j \mid j \in J'_1\} \cup \{g_{2q+5}\}$ and $A_2 = M \setminus A_1$ be the two agents' bundles respectively.
From agent~$1$'s perspective, agent~$1$'s bundle has utility $$((q+3)K + 2K) + (q+2)(4K) = (5q+13)K,$$ agent $2$'s bundle has utility $$((q+1)K + 2K) + (q+4)(4K) = (5q+19)K,$$ and a most valuable good in agent~$2$'s bundle (e.g., $g_{2q+6}$) has utility $6K$, so agent $1$ is EF1 towards agent $2$.
From agent~$2$'s perspective, agent~$2$'s bundle has utility $(q+1)K + 2K = (q+3)K$, agent~$1$'s bundle has utility $(q+3)K + 2K = (q+5)K$, and a most valuable good in agent $1$'s bundle (e.g., $g_{2q+5}$) has utility $2K$, so agent $2$ is EF1 towards agent $1$.
Accordingly, $(A_1, A_2)$ is an EF1 allocation with size vector $(q+2, q+4)$.

$(\Rightarrow)$ Let $(A_1, A_2)$ be an EF1 allocation with size vector~$\vec{s}$.
From agent $1$'s perspective, $u_1(M) = (10q+32)K$ and a most valuable good (e.g., $g_{2q+5}$) has utility $6K$.
For agent $1$ to be EF1 towards agent~$2$, we must have $$u_1(A_1) \geq \frac{(10q+32)K - 6K}{2} = (5q+13)K$$ and $$u_2(A_1) = u_1(A_1) - (q+2)(4K) \geq (q+5)K.$$
This means that $$u_1(A_2) = u_1(M) - u_1(A_1) \leq (5q+19)K$$ and $$u_2(A_2) = u_1(A_2) - (q+4)(4K) \leq (q+3)K.$$
On the other hand, from agent~$2$'s perspective, $u_2(M) = (2q+8)K$ and a most valuable good has utility $2K$.
For agent $2$ to be EF1 towards agent~$1$, we must have $$u_2(A_2) \geq \frac{(2q+8)K - 2K}{2} = (q+3)K$$ and $$u_1(A_2) = u_2(A_2) + (q+4)(4K) \geq (5q+19)K.$$
This means that $$u_2(A_1) = u_2(M) - u_2(A_2) \leq (q+5)K$$ and $$u_1(A_1) = u_2(A_1) + (q+2)(4K) \leq (5q+13)K.$$
By combining these inequalities, we conclude that these inequalities are tight, i.e., agent $1$'s utilities for both agents' bundles are \emph{exactly} $(5q+13)K$ and $(5q+19)K$ respectively so that the sum is $(10q+32)K$, and agent $2$'s utilities for both agents' bundles are \emph{exactly} $(q+5)K$ and $(q+3)K$ respectively so that the sum is $(2q+8)K$.
Additionally, both agents must each have a most valuable good worth $6K$ and~$2K$ to them respectively.
Without loss of generality, we may assume that $g_{2q+5} \in A_1$ and $g_{2q+6} \in A_2$ (note that $g_{2q+1}$ is also a most valuable good, but we use $g_{2q+5}$ and $g_{2q+6}$ for simplicity).

Since $g_{2q+6} \in A_2$, there are $q+3$ goods in $A_2 \setminus \{g_{2q+6}\}$ and $u_2(A_2 \setminus \{g_{2q+6}\}) = (q+3)K - 2K = (q+1)K$.
These goods are chosen from $M_1 = \{g_1, \ldots, g_{2q+1}\}$ and $M_0 = \{g_{2q+2}, g_{2q+3}, g_{2q+4}\}$.
Recall from the construction that $u_2(g) > K$ for all $g \in M_1$, and $u_2(g) = 0$ for all $g \in M_0$.
If $A_2 \setminus \{g_{2q+6}\}$ contains at least $q+1$ goods from $M_1$, then $u_2(A_2 \setminus \{g_{2q+6}\}) > (q+1)K$, a contradiction.
Therefore, $A_2 \setminus \{g_{2q+6}\}$ contains at most $q$ goods from $M_1$, and at least $3$ goods from $M_0$.
Since $|M_0| = 3$, we must have $M_0 \subseteq A_2$.
Note that $u_2(M_0) = 0$, so $u_2((A_2 \setminus \{g_{2q+6}\}) \setminus M_0) = (q+1)K$.
Therefore, the $q$ goods from~$M_1$ in agent $2$'s bundle have a total utility of $(q+1)K$.
These goods, together with $g_{2q+2}$, induce the set $Y_2$ with cardinality $q+1$ and sum $(q+1)K$.
Then, $Y_1 = Y \setminus Y_2$ and $Y_2$ give a required partition of $Y$.
\end{proof}

The proof of \Cref{thm:exist_gen_two_nphard} suggests that the problem is NP-hard even when the sizes of the two agents' bundles differ by \emph{exactly two}.
Note that this problem is in P when the sizes of the agents' bundles differ by \emph{at most one} (in fact, every such instance is a Yes-instance by \Cref{prop:exist_balanced}).
For two agents with binary utilities, we shall show later that the decision problem is in P (see \Cref{thm:exist_bin_const}).

\subsection{Constant Number of Agents}

Next, we discuss the complexity of the decision problem for a \emph{constant number of agents}.
In this case, we devise a pseudopolynomial-time algorithm for deciding the existence of an EF1 allocation with a given size vector.
This algorithm uses dynamic programming to check for such an allocation.

\begin{restatable}{lemma}{lemexistgenconstpseudopoly} \label{lem:exist_gen_const_pseudopoly}
Let an instance with $n$ agents and a size vector be given, where $n$ is a constant. Suppose that the utility of each good is an integer, and let $R = \max_{i\in N} u_i(M)$. 
Then, there exists an algorithm running in time polynomial in $m$ and~$R$ that decides whether the instance admits an EF1 allocation with the size vector.
\end{restatable}

\begin{proof}
The algorithm uses dynamic programming.
We construct a table with $m$ columns and $L$ rows, where $L$ will be specified later.
The index of each row is represented by a tuple containing $a_{i,j}$, $b_{i,j}$, and $c_i$ for each $i, j \in N$, i.e., $(a_{1,1},a_{1,2}, \ldots, a_{n,n}, b_{1,1},b_{1,2}, \ldots, b_{n,n}, c_1, \ldots, c_n)$. 
The value of $a_{i,j}$ is the utility of agent $j$'s bundle from agent $i$'s perspective, i.e., $a_{i,j} = u_i(A_j)$; the value of $b_{i,j}$ is the utility of a most valuable good in agent $j$'s bundle from agent $i$'s perspective, i.e., $b_{i,j} = \max_{g \in A_j} u_i(g)$ (note that this value is zero if $A_j = \emptyset$); and the value of $c_i$ is the number of goods in agent $i$'s bundle. 
Note that $a_{i,j}, b_{i,j} \in \{0, \ldots, R\}$ and $c_i \in \{0, \ldots, m\}$, so there are $L = (R+1)^{2n^2}(m+1)^n$ rows, which is polynomial in $m$ and $R$.
An entry in column $q$ represents whether an allocation involving $\{g_1, \ldots, g_q\}$ is possible for the tuple representing the row, and is either positive or negative.

Initialize every entry to negative. Consider the $n$ possibilities of adding $g_1$ to each of the agents' bundles respectively, and set the corresponding entries in the first column of the table to positive. 
In particular, for each $j \in N$, the row represented by the tuple such that $a_{i,j} = b_{i,j} = u_i(g_1)$ and $c_j = 1$ for all $i \in N$, and zero for all other values in the tuple, has the entry (in the first column) set to positive.

Now, for each $q \in \{2, \ldots, m\}$ in ascending order, for each positive entry in column $q-1$, consider the $n$ possibilities of adding $g_q$ into each of the $n$ agents' bundles respectively, and set the corresponding entry for each of these possibilities in column $q$ to positive.
Once this procedure is done, consider all positive entries in column $m$. If some positive entry corresponds to an EF1 allocation with the required size vector, then the instance admits such an EF1 allocation; otherwise, no such allocation exists.

Since $n$ is a constant, the number of entries in the table is polynomial in $m$ and $R$. At each column, there is a polynomial number of rows with positive entries, and hence the update is polynomial. Finally, checking for a feasible EF1 allocation at the last column can also be done in polynomial time.
\end{proof}

We now move to \emph{polynomial-time} algorithms that determine the existence of an EF1 allocation with a given size vector.
Recall that such an algorithm exists for two agents with identical utilities (\Cref{thm:exist_iden_two}).
However, it turns out that such an algorithm does not exist for three or more agents with identical utilities, unless $\text{P} = \text{NP}$.
In particular, we establish the NP-hardness of the decision problem via a reduction from \textsc{Balanced Multi-Partition} with $p = 2$, an NP-hard problem by \Cref{prop:balanced_multi_partition}.

\begin{restatable}{theorem}{thmexistidenconstnphard} \label{thm:exist_iden_const_nphard}
\textsc{Reformability} is weakly NP-complete for $n \geq 3$ agents with identical utilities, where $n$ is a constant.
\end{restatable}

\begin{proof}
Clearly, this problem is in NP.
The ``weak'' aspect is demonstrated in \Cref{lem:exist_gen_const_pseudopoly}.
Therefore, it suffices to show that this problem is NP-hard.

To show NP-hardness, we shall reduce from the NP-hard problem \textsc{Balanced Multi-Partition} with $p = 2$ (see \Cref{prop:balanced_multi_partition}).
Let a \textsc{Balanced Multi-Partition} instance with $p = 2$ be given.
Define a fair division instance as follows.
There are $n \geq 3$ agents with identical utilities, and a set of goods $M = \{g_1, \ldots, g_{2q}, h_1, \ldots, h_n\}$ such that $u(g_j) = x_j$ for $j \in \{1, \ldots, 2q\}$ and $u(h_k) = (q+1)K$ for $k \in \{1, \ldots, n\}$.
The size vector $\vec{s}$ is such that $s_1 = s_2 = q+1$ and $s_k = 1$ for all $k \in \{3, \ldots, n\}$.
This reduction can be done in polynomial time.
We claim that there exists an EF1 allocation with size vector $\vec{s}$ in this instance if and only if $X$ can be partitioned into multisets $X_1, X_2$ of equal cardinalities and sums.

$(\Leftarrow)$ Let $(X_1, X_2)$ be such a partition.
Define an allocation such that agent $k$ receives $h_k$ for $k \in N$, agent $1$ additionally receives the $q$ goods corresponding to the integers in $X_1$, and agent $2$ additionally receives the $q$ goods corresponding to the integers in $X_2$.
We show that this allocation is EF1.
The utilities of agent $1$'s and agent $2$'s bundles are $2(q+1)K$ each, and the utilities of the other agents' bundles are $(q+1)K$ each, so agents $1$ and $2$ do not envy anyone else.
Therefore, it remains to check that agent $k$ is EF1 towards agents $1$ and $2$ for $k \in \{3, \ldots, n\}$.
Upon the removal of the single good $h_1$ (resp.~$h_2$) from agent~$1$'s (resp.~agent $2$'s) bundle, the remaining bundle has utility $(q+1)K$, so agent $k$ is EF1 towards agent~$1$ (resp.~agent $2$).
Therefore, the allocation is EF1, as desired.

$(\Rightarrow)$ Let $(A_1, \ldots, A_n)$ be an EF1 allocation with size vector $(q+1, q+1, 1, \ldots, 1)$.
If agent $1$'s bundle has at least two goods from $\{h_1, \ldots, h_n\}$, then her bundle without the most valuable good has utility more than $(q+1)K$ since her bundle also contains other goods with positive utility.
Agent~$3$, having a bundle of utility at most $(q+1)K$, will not be EF1 towards agent $1$, contradicting the assumption that the allocation is EF1.
Therefore, agent $1$'s bundle has at most one good from $\{h_1, \ldots, h_n\}$; likewise for agent $2$'s bundle. 
This means that every agent receives \emph{exactly} one good from $\{h_1, \ldots, h_n\}$. 
Having established this, agent $3$'s bundle has a utility of $(q+1)K$, and agent $3$ is EF1 towards agent $1$. 
This means that agent $1$'s bundle without a most valuable good (say, some $h_k$) must have utility at most $(q+1)K$.
The same argument can be used to show the same statement for agent $2$'s bundle.
This means that the goods $\{g_1, \ldots, g_{2q}\}$ must be divided between agents $1$ and $2$ with each agent receiving a utility of \emph{exactly} $(q+1)K$.
Such a division of $\{g_1, \ldots, g_{2q}\}$ induces a partition of~$X$ into two multisets of equal cardinalities and sums, as desired.
\end{proof}

Since the decision problem is NP-hard even for identical utilities, it must also be NP-hard for general utilities.
We now consider another class of utilities: binary utilities.
When there are $n$ agents, every good~$g$ belongs to one of $2^n$ types of goods represented by the vector $(u_1(g), \ldots, u_n(g))$.
For the purpose of determining whether an EF1 allocation exists, it suffices to consider different goods of the same type as \emph{indistinguishable}.
We say that two allocations are in the same equivalence class if the number of goods of each type that each agent has is the same in both allocations.
If $\mathcal{A}$ is an EF1 allocation, then all allocations in the same equivalence class as $\mathcal{A}$ are also EF1 and are reachable from~$\mathcal{A}$.
We shall proceed with a result which enumerates all (essentially equivalent) EF1 allocations in time polynomial in the number of goods, provided that the number of agents is a constant.

\begin{restatable}{lemma}{lemenumeratebinconst} \label{lem:enumerate_bin_const}
Let an instance with $n$ agents with binary utilities and a size vector be given, where $n$ is a constant.
Then, there exists an algorithm running in time polynomial in $m$ that outputs all equivalence classes of EF1 allocations with the size vector.
\end{restatable}

\begin{proof}
An agent's bundle can be represented by a $2^n$-vector where each component of the vector is the number of goods of that type in her bundle.
Since the number of goods of each type is an integer between $0$ and $m$, there are $m+1$ possible values for each component, and hence at most $(m+1)^{2^n}$ possible vectors to represent each agent's bundle.
Allocations in an equivalence class can be represented by an ordered collection of $n$ such vectors---one for each agent---and there are at most $((m+1)^{2^n})^n$ such collections.
Since $((m+1)^{2^n})^n$ is polynomial in $m$ whenever $n$ is a constant, there are at most a polynomial number of possible equivalence classes of allocations in the instance.
For each of these equivalence classes of allocations, we can check whether an allocation in the equivalence class is EF1 and has the required size vector in polynomial time, and output the equivalence class if so.
Therefore, the overall running time is polynomial in $m$, as claimed.
\end{proof}

\Cref{lem:enumerate_bin_const} implies that the decision problem can be solved efficiently for binary utilities.

\begin{restatable}{theorem}{thmexistbinconst} \label{thm:exist_bin_const}
\textsc{Reformability} is in P for a constant number of agents with binary utilities.
\end{restatable}

\begin{proof}
Use the algorithm as described in \Cref{lem:enumerate_bin_const} to enumerate all equivalence classes with an EF1 allocation with the size vector, and output ``Yes'' if and only if such an equivalence class is found.
Note that if some allocation in an equivalence class is EF1, then all allocations in the same equivalence class are also EF1.
\end{proof}

\subsection{General Number of Agents}

For any constant number of agents, the problem of determining the existence of an EF1 allocation with a given size vector is \emph{weakly} NP-hard by \Cref{thm:exist_iden_const_nphard} (even for identical utilities).
For a general number of agents, the pseudopolynomial-time algorithm as described in \Cref{lem:exist_gen_const_pseudopoly} does not work, since that algorithm is at least exponential in the number of agents.
Therefore, the decision problem for a general number of agents might not be \emph{weakly} NP-hard.
We show that the problem is indeed \emph{strongly} NP-hard, even for identical utilities, by a reduction from \textsc{$3$-Partition}, a strongly NP-hard problem \citep[p.~224]{GareyJo79}.

\begin{restatable}{theorem}{thmexistidengennphard} \label{thm:exist_iden_gen_nphard}
\textsc{Reformability} is strongly NP-complete for identical utilities.
\end{restatable}

\begin{proof}
Clearly, this problem is in NP.
Therefore, it suffices to show that it is strongly NP-hard.

To this end, we shall reduce from \textsc{$3$-Partition}.
In \textsc{$3$-Partition}, we are given positive integers $q$ and $K$, and a multiset $X = \{x_1, \ldots, x_{3q}\}$ of positive integers of total sum $qK$.
The problem is to decide whether $X$ can be partitioned into multisets $X_1, \ldots, X_q$ of equal cardinalities and sums, i.e., for each $i \in \{1, \ldots, q\}$, $|X_i| = 3$ and the sum of all the integers in $X_i$ is exactly $K$.
This decision problem is known to be strongly NP-hard, even if $K/4 < x_j < K/2$ for every $j \in \{1, \ldots, 3q\}$ \citep[p.~224]{GareyJo79}.

Let an instance of \textsc{$3$-Partition} be given.
Define a fair division instance as follows.
There are $n = q+1$ agents with identical utilities, and a set of goods $M = \{g_1, \ldots, g_{3q+6}\}$ such that $u(g_j) = x_j$ for $j \in \{1, \ldots, 3q\}$ and $u(g_k) = K/5$ for $k \in \{3q+1, \ldots, 3q+6\}$.
The size vector $\vec{s}$ is such that $s_j = 3$ for all $j \in \{1, \ldots, q\}$ and $s_{q+1} = 6$.
This reduction can be done in polynomial time.
We claim that there exists an EF1 allocation with size vector $\vec{s}$ in this instance if and only if there exists a partition of $X$ into multisets $X_1, \ldots, X_q$ of equal cardinalities and sums.

$(\Leftarrow)$ Let such a partition be given. 
Define an allocation such that each of the first $q$ agents receives the three goods corresponding to each multiset, and agent $q + 1$ receives $\{g_{3q+1}, \ldots, g_{3q+6}\}$. 
Note that agents $1$ to $q$ have bundles worth $K$ each, and agent $q+1$ has a bundle worth $6K/5$. 
Since each of the goods from agent $(q+1)$'s bundle is worth exactly $K/5$, every agent is EF1 towards agent $q+1$. Accordingly, the allocation is EF1.

$(\Rightarrow)$ Let an EF1 allocation with size vector $\vec{s}$ be given. 
Note that every good is worth at least $K/5$, so agent $(q+1)$'s bundle without a most valuable good is worth at least $K$. 
If agent $q+1$ receives a bundle worth more than $6K/5$, then some agent receives a bundle worth less than $K$, and will not be EF1 towards agent $q+1$. 
Therefore, agent $q+1$ must receive a bundle worth at most $6K/5$. 
The only way this is possible is when agent $q+1$ receives $\{g_{3q+1}, \ldots, g_{3q+6}\}$. 
Now, since agents $1$ to $q$ are EF1 towards agent $q+1$, these agents must each receive a bundle worth at least $K$. The only way this is possible is when each of them receives a bundle worth \emph{exactly} $K$. 
This induces the desired partition.
\end{proof}

The reduction in \Cref{thm:exist_iden_gen_nphard} requires us to check whether the set of agents holding the goods with utilities equal to the integers in the partition problem is EF1 towards the remaining agents.
This stands in contrast to the reduction in \Cref{thm:exist_iden_const_nphard}, which entails checking whether the complementary set of agents is EF1 towards the remaining agents.

For binary utilities, the decision problem for a \emph{constant} number of agents is in P (\Cref{thm:exist_bin_const}).
The crucial reason is that in this case, the number of different types of goods is also a constant, which allows us to enumerate all the (essentially equivalent) EF1 allocations in polynomial time (\Cref{lem:enumerate_bin_const}).
This is no longer possible when the number of agents is non-constant.
In fact, we show that the decision problem is NP-hard for a general number of agents with binary utilities.
To this end, we reduce from \textsc{Graph $k$-Colorability}, which is NP-hard for any fixed $k \geq 3$ \citep[p.~191]{GareyJo79}.

\begin{restatable}{theorem}{thmexistbingennphard} \label{thm:exist_bin_gen_nphard}
\textsc{Reformability} is NP-complete for binary utilities.
\end{restatable}

\begin{proof}
Clearly, this problem is in NP.
Therefore, it suffices to show that it is NP-hard.

To this end, we shall reduce from \textsc{Graph $k$-Colorability} with $k \geq 3$.
In \textsc{Graph $k$-Colorability}, we are given a graph $G = (V, E)$ and a positive integer $k$, and the problem is to decide whether $G$ is $k$-colorable, i.e., whether each of the vertices in $V$ can be assigned one of $k$ colors in such a way that no two adjacent vertices are assigned the same color.
This decision problem is known to be NP-hard for any fixed $k \geq 3$ \citep[p.~191]{GareyJo79}.

Let an instance of \textsc{Graph $k$-Colorability} be given with a fixed $k \geq 3$, where $V = \{v_1, \ldots, v_p\}$ and $E = \{e_1, \ldots, e_q\}$. 
Define a fair division instance as follows.
There are $n = q+k$ agents where the first $q$ agents are called \emph{edge agents} and the last $k$ agents are called \emph{color agents}.
There are $m = kp$ goods.
Each color agent assigns zero utility to every good.
For $r \in \{1, \ldots, q\}$, if $e_r = \{v_i, v_j\}$, then the $r^\text{th}$ edge agent assigns a utility of $1$ to each of $g_i$ and $g_j$, and zero utility to every other good.
Note that only the first $p$ goods correspond to vertices and are valuable to the edge agents whose corresponding edges are incident to the vertices; the remaining $(k-1)p$ goods are not valuable to any agent.
The size vector $\vec{s}$ is such that $s_r = 0$ for each edge agent $r$ and $s_c = p$ for each color agent $c$.
This reduction can be done in polynomial time.
We claim that there exists an EF1 allocation with size vector $\vec{s}$ in this instance if and only if $G$ is $k$-colorable.

$(\Leftarrow)$ Let a proper $k$-coloring of $G$ be given.
For $t \in \{1, \ldots, p\}$, if vertex $v_t$ is assigned the color~$c$, then allocate $g_t$ to the color agent~$c$.
Since there are $p$ such goods and each color agent is supposed to have $p$ goods in her bundle, it is possible to allocate all of these goods.
Subsequently, allocate the remaining goods arbitrarily among the color agents until every color agent has exactly $p$ goods.
We claim that this allocation is EF1.
Every color agent assigns zero utility to every good and is thus EF1 towards every other agent.
Each edge agent assigns a utility of $1$ to only two goods, so we only need to check that these two goods are in different bundles.
Indeed, these two goods correspond to vertices which are adjacent to each other in $G$, and proper coloring of $G$ implies that the vertices are of different colors, so the corresponding goods are in different color agents' bundles.
Hence, the allocation is EF1.

$(\Rightarrow)$ Let an EF1 allocation with size vector $\vec{s}$ be given.
For $t \in \{1, \ldots, p\}$, if the good $g_t$ is with color agent $c$, assign $v_t$ to color $c$.
We claim that this coloring is a proper $k$-coloring of $G$.
Since there are $k$ color agents, at most $k$ colors are used.
Therefore, it suffices to check that adjacent vertices are assigned different colors.
Let $v_i, v_j \in V$ be adjacent vertices.
Then, there exists an edge $e_r = \{v_i, v_j\}$.
The edge agent $r$ assigns a utility of $1$ to each of $g_i$ and~$g_j$.
Since agent $r$'s bundle is empty, $g_i$ and $g_j$ must be in different (color agents') bundles in order for agent $r$ to be EF1 towards every other agent.
This implies that $v_i$ and~$v_j$ are assigned different colors.
\end{proof}

Even though the decision problem is NP-hard for identical \emph{or} binary utilities, we prove next that it can be solved in polynomial time for identical \emph{and} binary utilities.
Indeed, this can be done by checking whether the total number of valuable goods is within a certain threshold which can be computed in polynomial time.
This threshold is in fact the sum over all agents $j\in N$ of the minimum between $s_j$ and $1 + \min_{i \in N} s_i$.
If the number of valuable goods is within this threshold, then we can distribute the valuable goods in a round-robin fashion first, thereby ensuring EF1.
Conversely, if the number of valuable goods is beyond this threshold, then some agent $i$ with the minimum $s_i$ will not be EF1 towards another agent who receives at least $s_i + 2$ valuable goods.

\begin{restatable}{theorem}{thmexistidenbingen} \label{thm:exist_idenbin_gen}
\textsc{Reformability} is in P for identical binary utilities.
\end{restatable}

\begin{proof}
Let $\vec{s}$ be the given size vector.
Let $s_0 = \min_{i\in N} s_i$ be the size of the smallest bundle, and $n_0 = |\{i \in N \mid s_i = s_0\}|$ be the number of agents with exactly $s_0$ goods in their bundles.
We claim that an EF1 allocation with size vector $\vec{s}$ exists if and only if the number of valuable goods is at most $s_0n + n - n_0$.
Note that checking whether the number of valuable goods is at most $s_0n + n - n_0$ can be done in polynomial time.

If an EF1 allocation with size vector $\vec{s}$ exists, then an agent with bundle size $s_0$ receives at most $s_0$ valuable goods.
For this agent to be EF1 towards every other agent, every other agent can only receive at most $s_0 + 1$ valuable goods.
Since there are $n - n_0$ agents with bundle size at least $s_0 + 1$ and $n_0$ agents with bundle size exactly $s_0$, the total number of valuable goods is at most $(n - n_0)(s_0 + 1) + n_0s_0 = s_0n + n - n_0$.

Conversely, if the number of valuable goods is at most $s_0n + n - n_0$, then we can allocate the valuable goods in a round-robin fashion up to the bundle size of each agent, followed by the non-valuable goods.
Note that every agent's bundle size is at least $s_0$.
If there are at most $s_0n$ valuable goods, then these valuable goods can be distributed fairly with the difference in the number of valuable goods between agents being at most one, and hence the allocation is EF1.
Otherwise, the first $s_0n$ valuable goods can be distributed so that every agent receives $s_0$ of them.
Since there are $n - n_0$ agents with bundle size at least $s_0 + 1$ and at most $n - n_0$ valuable goods left, the remaining valuable goods can be arbitrarily allocated to these agents so that each of these agents receives at most one more valuable good.
Then, each agent receives $s_0$ or $s_0 + 1$ valuable goods, and so the allocation is EF1.
\end{proof}

\section{Optimal Number of Exchanges} \label{sec:optimal}

In this section, we consider the complexity of computing the \emph{optimal number of exchanges} required to reach an EF1 allocation from an initial allocation.

Recall that the decision problem in \Cref{sec:exist} is to determine whether there exists an EF1 allocation that can be reached from a given initial allocation.
This is equivalent to determining whether the optimal number of exchanges to reach an EF1 allocation is finite or not.
We have established a few scenarios where there exist polynomial-time algorithms for this task: (a) two agents with identical utilities (\Cref{thm:exist_iden_two}), (b) any constant number of agents with binary utilities (\Cref{thm:exist_bin_const}), and (c) any number of agents with identical binary utilities (\Cref{thm:exist_idenbin_gen}).
For these scenarios, we can run the respective polynomial-time algorithms to decide whether such an EF1 allocation exists---if none exists, then the optimal number of exchanges is $\infty$.
Therefore, for the proofs in this section pertaining to these scenarios, we proceed with the assumption that such an EF1 allocation exists.
We will show that the problem of computing the optimal number of exchanges for these scenarios is also in P; our algorithms can be modified to compute an optimal \emph{sequence} of exchanges as well.

For the remaining scenarios where the decision problem in \Cref{sec:exist} is NP-hard, it is NP-hard to even decide whether the optimal number of exchanges to reach an EF1 allocation is finite or not.
Therefore, for these scenarios, we shall focus on the special case where the given size vector is \emph{balanced}, so that the optimal number of exchanges is guaranteed to be finite (see \Cref{prop:exist_balanced}).
Even with this assumption, we will show that the computational problem for these scenarios remains NP-hard.

For convenience, we refer to as \textsc{Optimal Exchanges} the problem of deciding---given an instance, an initial allocation in the instance, and a number $k$---whether the optimal number of exchanges required to reach an EF1 allocation is at most $k$.
Note that the optimal number of exchanges in the scenarios mentioned above are finite.
By the proof of \Cref{prop:reachable}, the optimal number of exchanges in these scenarios is polynomial in the number of agents and the number of goods.
As a result, \textsc{Optimal Exchanges} is in NP regardless of the number of agents, as we can easily verify whether an exchange path starting from the given initial allocation indeed reaches an EF1 allocation using at most $k$ exchanges.

\subsection{Two Agents}

We begin with the case of two agents.
For two agents with identical utilities, we show that there exists a polynomial-time algorithm that computes the optimal number of exchanges.
This algorithm performs the exchanges until an EF1 allocation is reached, while keeping track of the number of exchanges required.
The algorithm is ``greedy'' in the sense that at each step, it performs an exchange involving the most valuable good from the agent whose bundle has the higher utility, and the least valuable good from the other agent.
We demonstrate that this choice is the best for the number of exchanges required to reach an EF1 allocation.

\begin{restatable}{theorem}{thmoptimalidentwo} \label{thm:optimal_iden_two}
\textsc{Optimal Exchanges} is in P for two agents with identical utilities.
\end{restatable}

\begin{proof}
We show that we can compute the optimal number of exchanges in polynomial time.
We assume that an EF1 allocation with the given size vector exists.
If the initial allocation~$\mathcal{A}$ is EF1, we are done. 
Otherwise, assume without loss of generality that agent $2$ has a higher utility than agent $1$ in $\mathcal{A}$.
Let $\vec{s} = (s_1, s_2)$ be the size vector.
By rearranging the labels of the goods, assume that the goods are in non-increasing order of utility, i.e., $u(g_1) \geq u(g_2) \geq \cdots \geq u(g_m)$.
The algorithm proceeds as follows: repeatedly exchange the most valuable good in agent $2$'s bundle with the least valuable good in agent~$1$'s bundle until agent $1$ is EF1 towards agent $2$.
The optimal number of exchanges required is then the number of exchanges made in this algorithm.

First, we claim that in each exchange, a good from $\{g_1, \ldots, g_{s_1}\}$ in agent $2$'s bundle is always exchanged with a good from $\{g_{s_1+1}, \ldots, g_m\}$ in agent $1$'s bundle.
Suppose on the contrary that this is not true.
Let $\mathcal{A}'$ be the allocation just before we make the exchange that violates this claim.
The only way for the claim to be violated is when $A'_1 = \{g_1, \ldots, g_{s_1}\}$ and $A'_2 = \{g_{s_1+1}, \ldots, g_m\}$.
If $s_1 \leq s_2$, then by \Cref{lem:exist_iden_two}, there does not exist an EF1 allocation with size vector $\vec{s}$---this would contradict our assumption that an EF1 allocation with size vector $\vec{s}$ exists.
Otherwise, $s_1 > s_2$, and every good in $A'_1$ has a higher utility than every good in $A'_2$, so agent $1$ is EF1 towards agent $2$.
This would contradict our assumption that the algorithm has not terminated.
Hence, the claim is true.

By the claim in the previous paragraph, the total number of exchanges made is at most $\min \{s_1, s_2\} \leq m$.
Each exchange can be performed in polynomial time, and so the algorithm terminates in polynomial time.
We show next that an EF1 allocation is obtained when the algorithm terminates.
It suffices to show that agent $2$ is EF1 towards agent $1$ in the final allocation.
To this end, we show that agent $2$ is EF1 towards agent~$1$ at every step of the algorithm.
Let the initial allocation be $\mathcal{A}^0 = \mathcal{A}$, and let $\mathcal{A}^t$ be the allocation after $t$ steps of the algorithm.
Note that $\mathcal{A}^0$ satisfies the condition that agent $2$ is EF1 towards agent $1$, since agent $2$ has a higher utility than agent $1$ in $\mathcal{A}$.
We show that if $\mathcal{A}^t$ has the property that agent $2$ is EF1 towards agent $1$ and agent $1$ is \emph{not} EF1 towards agent $2$, then $\mathcal{A}^{t+1}$ has the property that agent $2$ is EF1 towards agent $1$.
Suppose that $g \in A^t_2$ is exchanged with $h \in A^t_1$. 
Since agent~$1$ is not EF1 towards agent~$2$ in ${\cal A}^t$, it holds that 
$u(A^t_1) < u(A^t_2 \setminus \{g\})$. 
Thus, 
\begin{align*}
    u(A^{t+1}_2) &\geq u(A^{t+1}_2 \setminus \{h\}) 
    = u(A^t_2 \setminus \{g\}) 
    > u(A^t_1) 
    \geq u(A^t_1 \setminus \{h\}) 
    = u(A^{t+1}_1 \setminus \{g\}),
\end{align*}
showing that agent $2$ is EF1 towards agent $1$ in $\mathcal{A}^{t+1}$.

Finally, we show that the optimal number of exchanges required to reach an EF1 allocation is at least the number of exchanges made in this algorithm.
Let $T$ be the number of exchanges made in this algorithm.
For each $t \in \{1, \ldots, T\}$, let $g^t \in A_2$ (resp.~$h^t \in A_1$) be the good in agent $2$'s bundle (resp.~agent $1$'s bundle) that is exchanged at the $t^\text{th}$ step of the algorithm.
Note that $u(g^1) \geq \cdots \geq u(g^T) \geq u(h^T) \geq \cdots \geq u(h^1)$.
Also, we have $u(A^{T-1}_1) < u(A^{T-1}_2 \setminus \{g^T\})$, where $g^T$ is a good with the highest utility in agent $2$'s bundle in $\mathcal{A}^{T-1}$.
Suppose on the contrary that only $k \leq T-1$ exchanges are required to reach an EF1 allocation.
Since $\mathcal{A}$ is not EF1, we have $1 \leq k < T$.
Let $(B_1, B_2)$ be the EF1 allocation after the $k$ exchanges.
The utility of $B_1$ is upper-bounded by the utility of $A_1$ after adding $k$ goods of the highest utility from $A_2$ and removing $k$ goods of the lowest utility from $A_1$, so
\begin{align*}
    u(B_1) &\leq u((A_1 \cup \{g^1, \ldots, g^k\}) \setminus \{h^1, \ldots, h^k\}) \\
    &= u(A_1) + \sum_{t=1}^k (u(g^t) - u(h^t)) \\
    &\leq u(A_1) + \sum_{t=1}^{T-1} (u(g^t) - u(h^t)) \\
    &= u((A_1 \cup \{g^1, \ldots, g^{T-1}\}) \setminus \{h^1, \ldots, h^{T-1}\}) = u(A^{T-1}_1).
\end{align*}
On the other hand, the utility of $B_2$ without the most valuable good is lower-bounded by the utility of~$A_2$ after adding $k$ goods of the lowest utility from $A_1$ and removing $k+1 \leq T$ goods of the highest utility from $A_2$, so we have
\begin{align*}
    u(B_2 \setminus \{g\}) &\geq u((A_2 \cup \{h^1, \ldots, h^k\}) \setminus \{g^1, \ldots, g^{k+1}\}) \\
    &= u(A_2) - u(g^1) - \sum_{t=1}^k (u(g^{t+1}) - u(h^t)) \\
    &\geq u(A_2) - u(g^1) - \sum_{t=1}^{T-1} (u(g^{t+1}) - u(h^t)) \\
    &= u((A_2 \cup \{h^1, \ldots, h^{T-1}\}) \setminus \{g^1, \ldots, g^T\}) 
    = u(A^{T-1}_2 \setminus \{g^T\})
\end{align*}
for every $g \in B_2$.
This implies that $$u(B_1) \leq u(A^{T-1}_1) < u(A^{T-1}_2 \setminus \{g^T\}) \leq u(B_2 \setminus \{g\})$$ for all $g \in B_2$.
Hence, agent $1$ is not EF1 towards agent $2$ in $(B_1, B_2)$, contradicting the assumption that $(B_1, B_2)$ is EF1.
It follows that at least $T$ exchanges are required to reach an EF1 allocation.
\end{proof}

However, if the utilities are not identical, then computing the optimal number of exchanges is NP-hard, even for balanced allocations.
We show this by modifying the construction from the NP-hardness proof of \Cref{thm:exist_gen_two_nphard} in determining the existence of an EF1 allocation with a given size vector.

\begin{restatable}{theorem}{thmoptimalgentwonphard} \label{thm:optimal_gen_two_nphard}
\textsc{Optimal Exchanges} is NP-complete for two agents, even when the initial allocation is balanced.
\end{restatable}

\begin{proof}
Clearly, this problem is in NP.
To demonstrate NP-hardness, we modify the construction from the proof of \Cref{thm:exist_gen_two_nphard}.
Recall that we have $Y = \{y_1, \ldots, y_{2q+2}\}$ with $K < y_j \leq 2K$ for $j \in \{1, \ldots, 2q\}$, $y_{2q+1} = 2K$, and $y_{2q+2} = 0$.
A fair division instance~$\mathcal{I}'$ is defined with $n = 2$ agents and a set of goods $M' = \{g_1, \ldots, g_{2q+6}\}$.
Agent $2$'s utility is such that $u_2(g_j) = y_j$ for $j \in \{1, \ldots, 2q+2\}$, $u_2(g_{2q+3}) = u_2(g_{2q+4}) = 0$, and $u_2(g_{2q+5}) = u_2(g_{2q+6}) = 2K$.
Agent $1$'s utility is such that $u_1(g) = u_2(g) + 4K$ for all $g \in M$.
The size vector $\vec{s}\,'$ is $(q+2, q+4)$.
In \Cref{thm:exist_gen_two_nphard}, it was proven that there exists an EF1 allocation with size vector $\vec{s}\,'$ in this instance if and only if $Y$ can be partitioned into two multisets $Y_1$ and $Y_2$ of equal cardinalities (i.e., of size $q+1$ each) with sums $(q+3)K$ and $(q+1)K$ respectively.
Both problems were proven to be NP-hard.

Define a new fair division instance $\mathcal{I}$ as follows.
There are $n = 2$ agents and a set of goods $M = \{g_1, \ldots, g_{4q+12}\}$.
For $j \in \{1, \ldots, 2q+6\}$, the utility of $g_j$ for each agent is identical to that in the original fair division instance $\mathcal{I}'$.
For $j \in \{2q+7, \ldots, 4q+12\}$, we have $u_i(g_j) = 0$ for $i \in \{1, 2\}$.
The size vector $\vec{s}$ is $(2q+6, 2q+6)$.
In the initial allocation $\mathcal{A}$, agent $1$ has $A_1 = \{g_{2q+7}, \ldots, g_{4q+12}\}$ and agent $2$ has $A_2 = \{g_1, \ldots, g_{2q+6}\}$.
This reduction can be done in polynomial time.
We claim that the optimal number of exchanges required to reach an EF1 allocation from $\mathcal{A}$ is at most $q+2$ in $\mathcal{I}$ if and only if there exists an EF1 allocation with size vector $\vec{s}\,'$ in $\mathcal{I}'$.

$(\Leftarrow)$ Let $(A'_1, A'_2)$ be an EF1 allocation with size vector $\vec{s}\,'$ in $\mathcal{I}'$.
Note that $A'_1 \subseteq A_2$ and $|A'_1| = q+2$.
In $\mathcal{A}$, exchange the $q+2$ goods from $A'_1$ with any $q+2$ goods in $A_1$.
This requires a total of $q+2$ exchanges.
The new allocation has exactly the same goods as that in $(A'_1, A'_2)$ along with other goods with zero utility, so it is EF1.
Therefore, the optimal number of exchanges to reach an EF1 allocation from $\mathcal{A}$ is at most $q+2$.

$(\Rightarrow)$ Suppose that an EF1 allocation $\mathcal{B}$ is reached from $\mathcal{A}$ after $t \leq q+2$ exchanges in $\mathcal{I}$.
We may assume that every good is not exchanged more than once.
By the same reasoning as in the proof of \Cref{thm:exist_gen_two_nphard}, we must have $u_1(B_1) \geq (5q+13)K$ and $u_2(B_1) \leq (q+5)K$.
Since $t$ goods are transferred from $A_2$, we have $$u_1(B_1) = u_2(B_1) + t(4K) \leq (q+5)K + (q+2)(4K) = (5q+13)K.$$
This means that the inequalities for $u_1(B_1)$ are tight, and we have $u_1(B_1) = 5q+13$ and $t = q+2$.
Letting $M_1 = A_2 \cap B_1$, we have $|M_1| = q+2$ and $M_1 \subseteq A_2 \subseteq M'$.
Since $(B_1, B_2)$ is an EF1 allocation, the allocation that removes all goods with zero utility is also EF1, namely, $(M_1, M' \setminus M_1)$.
This induces an EF1 allocation with size vector $\vec{s}\,'$ in $\mathcal{I}'$.
\end{proof}

\subsection{Constant Number of Agents}

While a polynomial-time algorithm to compute the optimal number of exchanges exists for two agents with identical utilities, this is not the case for three or more agents unless $\text{P} = \text{NP}$.
Indeed, we establish the NP-hardness of this problem via a reduction from \textsc{Balanced Multi-Partition} with $p \geq 2$, an NP-hard problem by \Cref{prop:balanced_multi_partition}.

\begin{restatable}{theorem}{thmoptimalidenconst} \label{thm:optimal_iden_const}
\textsc{Optimal Exchanges} is NP-complete for $n \geq 3$ agents with identical utilities, where $n$ is a constant, even when the initial allocation is balanced.
\end{restatable}

\begin{proof}
Clearly, this problem is in NP.
To demonstrate NP-hardness, we shall reduce from the NP-hard problem \textsc{Balanced Multi-Partition} (see \Cref{prop:balanced_multi_partition}).
Let a \textsc{Balanced Multi-Partition} instance be given with $p \geq 2$.
Define a fair division instance as follows.
There are $n = p + 1$ agents with identical utilities, and a set of goods $M = \{g_1, \ldots, g_{n(pq+q+2)}\}$ such that $u(g_j) = x_j$ for all $j \in \{1, \ldots, pq\}$, $u(g_j) = K$ for all $j \in \{pq+1, \ldots, pq+q+2\}$, and $u(g) = 0$ for the remaining goods $g$.
Note that every good in $M_1 = \{g_1, \ldots, g_{pq}\}$ has utility more than $K$ and at most $2K$, every good in $M_2 = \{g_{pq+1}, \ldots, g_{pq+q+2}\}$ has utility exactly $K$, and every good in $M \setminus (M_1 \cup M_2)$ has zero utility.
The size vector $\vec{s}$ is such that $s_i = pq+q+2$ for all $i \in N$.
In the initial allocation $\mathcal{A}$, agent~$n$ has $M_1 \cup M_2$, while the remaining agents have the remaining goods.
This reduction can be done in polynomial time.
We claim that the optimal number of exchanges required to reach an EF1 allocation from $\mathcal{A}$ is at most $pq$ if and only if $X$ can be partitioned into multisets $X_1, \ldots, X_p$ of equal cardinalities and sums.

$(\Leftarrow)$ Let $(X_1, \ldots, X_p)$ be such a partition.
For each $j \in \{1, \ldots, pq\}$, if $x_j$ is in $X_i$ for some $i \in \{1, \ldots, p\}$, exchange $g_j$ in agent $n$'s bundle with a zero-utility good in agent $i$'s bundle.
After $pq$ exchanges, each agent $i \in \{1, \ldots, p\}$ has $q$ goods corresponding to the integers in $X_i$ along with other goods with zero utility, and the total utility of these goods is $(q+1)K$.
Meanwhile, agent $n$ has $M_2$ along with other goods with zero utility.
There are $q+2$ goods with utility $K$ each, and so the utility of agent $n$'s bundle without a most valuable good is $(q+1)K$.
This shows that the resulting allocation is EF1.
Therefore, the optimal number of exchanges required to reach an EF1 allocation from $\mathcal{A}$ is at most~$pq$.

$(\Rightarrow)$ Suppose that the optimal number of exchanges required to reach an EF1 allocation from $\mathcal{A}$ is at most $pq$.
Let $\mathcal{B}$ be one such EF1 allocation after these exchanges.
Since at most $pq$ exchanges were made, agent $n$ has at least $|M_1 \cup M_2| - pq = q+2$ goods from $M_1 \cup M_2$ in $\mathcal{B}$.
Every good in $M_1 \cup M_2$ has utility at least $K$, so the utility of agent $n$'s bundle in $\mathcal{B}$ without a most valuable good is at least $(q+1)K$.
For every agent to be EF1 towards agent $n$, they must each have a utility of at least $(q+1)K$ in $\mathcal{B}$.
The total utility of agent $1$ to agent $p$'s bundles is therefore at least $p(q+1)K$, which is exactly the utility of $M_1$.
Since every good in $M_1$ has a higher utility than every good in $M_2$, this implies that the only possibility is that all the goods in $M_1$ go to agents $1$ to $p$, leaving $M_2$ (along with other goods with zero utility) with agent $n$.
This means that agent $n$'s bundle in $\mathcal{B}$ without a most valuable good has utility \emph{exactly} $(q+1)K$, and that the goods in $M_1$ must be split among agents $1$ to $p$ so that every agent receives a utility of $(q+1)K$ from $M_1$.
Furthermore, none of these agents can receive more than $q$ goods from~$M_1$; otherwise, if some agent receives at least $q+1$ goods from $M_1$, then the utility of her bundle is more than $(q+1)K$, which leaves another agent with utility less than $(q+1)K$, a contradiction.
Therefore, agent $1$ to $p$ each receives exactly $q$ goods from $M_1$.
Hence, it is possible to partition the goods in $M_1$ into $p$ bundles so that each bundle has exactly $q$ goods and the utility of each bundle is exactly $(q+1)K$.
This induces a partition of $X$ into $p$ multisets of equal cardinalities and sums. 
\end{proof}

Next, we consider binary utilities.
We have shown that deciding whether the optimal number of exchanges to reach an EF1 allocation is finite can be done in polynomial time (\Cref{thm:exist_bin_const}).
We now demonstrate that the same is true for computing this exact number.

\begin{restatable}{theorem}{thmoptimalbinconst} \label{thm:optimal_bin_const}
\textsc{Optimal Exchanges} is in P for a constant number of agents with binary utilities.
\end{restatable}

Before we establish \Cref{thm:optimal_bin_const}, we first prove the following lemma, which shows that finding the optimal number of exchanges between two equivalence classes of allocations can be done efficiently.

\begin{restatable}{lemma}{lemoptimalbinconst} \label{lem:optimal_bin_const}
Let an instance with $n$ agents with binary utilities be given, where $n$ is a constant. 
Then, there exists an algorithm running in time polynomial in $m$ that computes the optimal number of exchanges required to reach some allocation in a given equivalence class from another given allocation.
\end{restatable}

\begin{proof}
If the size vectors of the given allocation and an allocation in the given equivalence class are different, then the optimal number of exchanges is $\infty$.
Therefore, we may henceforth assume that the size vectors are the same.

For each agent $i$ and each type of good $c$ in the given initial allocation, we have an $n$-vector such that the $j^\text{th}$ component of this vector represents the number of goods of type $c$ that need to be moved from agent $i$ in the initial allocation to agent $j$ in some final allocation in the given equivalence class under some exchange.
Since the number of goods of each type is an integer between $0$ and $m$, there are $m+1$ possible numbers for each component, and hence at most $(m+1)^n$ possible vectors to represent this information.
The movement of goods from the initial allocation to the final allocation can be represented by an ordered collection of such vectors over all agents and over all types of goods in each agent's initial allocation.
Since there are $n$ agents and $2^n$ types of goods, there are at most $((m+1)^n)^{n \cdot 2^n}$ such collections.
Since $((m+1)^n)^{n \cdot 2^n}$ is polynomial in $m$ whenever $n$ is a constant, there are at most a polynomial number of such movements between the two allocations.
For each of these movements, we can verify in polynomial time whether it indeed gives some final allocation in the equivalence class.
We thus have an enumeration of all such feasible movements in polynomial time.

For each feasible movement between the initial allocation $\mathcal{A}$ and the final allocation $\mathcal{B}$, define a directed graph where the vertices are the agents and each edge $e_g$ represents a good such that if $g \in A_i \cap B_j$, then $e_g = (i, j)$.
\citet[Prop.~4.1]{IgarashiKaSu24} showed that the number of exchanges required to reach $\mathcal{B}$ from $\mathcal{A}$ is $m - c^*$, where $c^*$ is the maximum number of disjoint cycles in the item exchange graph.
Note that each cycle consists of a subset of the agents in some order, so the number of cycle types, $L$, is at most $(n+1)!$, which is a constant.
We have an $L$-vector such that the $k^\text{th}$ component of this vector represents the number of cycles in the item exchange graph of cycle type $k$.
Since the number of cycles of each cycle type is an integer between $0$ and $m$, there are $m+1$ possible numbers for each component, and hence at most $(m+1)^L$ possible vectors to represent the cycles in the graph, which is polynomial in $m$.
We can enumerate all such vectors, consider only those vectors that represent the item exchange graph, and output the maximum number of disjoint cycles from these vectors as $c^*$ in polynomial time.
Then, calculating $m - c^*$ gives the optimal number of exchanges for that feasible movement of goods.
Finally, the minimum optimal number of exchanges across all feasible movements is the quantity we desire.
\end{proof}

We proceed to the proof of \Cref{thm:optimal_bin_const}.

\begin{proof}[Proof of \Cref{thm:optimal_bin_const}]
We use the algorithm in \Cref{lem:enumerate_bin_const} to enumerate all the possible equivalence classes of EF1 allocations with the given size vector in polynomial time.
For each equivalence class of allocations, we use the algorithm in \Cref{lem:optimal_bin_const} to compute the optimal number of exchanges required to reach some allocation from this equivalence class from the given initial allocation in polynomial time.
The smallest such number will then answer the decision problem of \textsc{Optimal Exchanges}.    
\end{proof}

\subsection{General Number of Agents}

Let us now consider a non-constant number of agents.
We have shown that computing the optimal number of exchanges required to reach an EF1 allocation is NP-hard, even for identical utilities (\Cref{thm:optimal_iden_const}).
We thus consider binary utilities.

Although this problem belongs to P when the number of agents is a constant (\Cref{thm:optimal_bin_const}), we show that it is NP-hard for a general number of agents, even for the special case where the initial allocation is balanced.
To this end, we reduce from \textsc{Exact Cover by 3-Sets (X3C)}, an NP-hard problem \citep[p.~221]{GareyJo79}.

\begin{restatable}{theorem}{thmoptimalbingen} \label{thm:optimal_bin_gen}
\textsc{Optimal Exchanges} is NP-complete for binary utilities, even when the initial allocation is balanced.
\end{restatable}

\begin{proof}
Clearly, this problem is in NP.
To demonstrate NP-hardness, we shall reduce from \textsc{Exact Cover by 3-Sets (X3C)}.
In \textsc{X3C}, we are given positive integers $p$ and $q$, a set $X = \{x_1, \ldots, x_{3q}\}$, and a collection $C = \{Y_1, \ldots, Y_p\}$ of three-element subsets of $X$, i.e., for each $j \in \{1, \ldots, p\}$, $|Y_j| = 3$ and $Y_j \subseteq X$.
The problem is to decide whether there exists an exact cover for $X$ in $C$, i.e., whether there exists $D \subseteq C$ such that $|D| = q$ and for all $x \in X$, there exists $Y \in D$ such that $x \in Y$.
This decision problem is known to be NP-hard \citep[p.~221]{GareyJo79}.

Let an \textsc{X3C} instance be given.
Note that if there exists an exact cover $D$ for $X$, and if some $x' \in X$ appears in exactly one $Y' \in C$, then $Y'$ must be in $D$, and moreover, the other two elements in $Y' \setminus \{x'\}$ must not appear in any other three-element sets in $D$.
In this case, we can reduce the problem further by considering the set $X \setminus Y'$ and the collection $\{Y_j \in C \mid Y_j \cap Y' = \emptyset\}$ instead.
Therefore, assume without loss of generality that each $x \in X$ appears in at least two three-element sets in $C$.

Define a fair division instance as follows.
There are $n = 3q+1$ agents with binary utilities, and a set of goods $M = \{g_{i, j} \mid i \in \{1, \ldots, 3q+1\}, j \in \{1, \ldots, p\}\}$.
For notational simplicity, let $h_j = g_{3q+1, j}$ for all $j \in \{1, \ldots, p\}$; the good $h_j$ is associated with $Y_j$.
The size vector is $\vec{s} = (p, \ldots, p)$.
In the initial allocation $\mathcal{A} = (A_1, \ldots, A_{3q+1})$, we have $A_i = \{g_{i, j} \mid j \in \{1, \ldots, p\}\}$ for all $i \in N$.
Agent $3q+1$ is a special agent and assigns zero utility to every good.
For each non-special agent $i \in \{1, \ldots, 3q\}$, let $n_i$ be the number of three-element subsets in $C$ that contain $x_i$, i.e., $n_i = |\{Y \in C \mid x_i \in Y\}|$.
By assumption, we have $n_i \geq 2$.
Then, agent $i$ values exactly $n_i - 2$ goods in $A_i$, e.g., $u_i(g_{i, j}) = 1$ for $j \in \{1, \ldots, n_i - 2\}$ and $u_i(g_{i, j}) = 0$ for $j \in \{n_i - 1, \ldots, p\}$.
Agent $i$ also values the goods associated with any $Y_j$ that contains $x_i$, i.e., $u_i(h_j) = 1$ if and only if $x_i \in Y_j$.
Agent $i$ assigns zero utility to every other good not mentioned.
Note that in the initial allocation, from each non-special agent $i$'s perspective, the utility of agent $i$'s bundle is $n_i - 2$, the utility of agent $(3q+1)$'s bundle is $n_i$, and the utilities of the other agents' bundles are zero.
This reduction can be done in polynomial time.
We claim that the optimal number of exchanges required to reach an EF1 allocation from $\mathcal{A}$ is at most $q$ if and only if there exists an exact cover for $X$ in~$C$.

$(\Leftarrow)$ Let $D$ be an exact cover for $X$.
For each $Y_j \in D$, we perform one exchange as follows: select any $x_i \in Y_j$ arbitrarily, and exchange $g_{i, p}$ in agent $i$'s bundle with $h_j$ in agent $(3q+1)$'s bundle.
Note that there are exactly $q$ exchanges, since $|D| = q$.
We claim that the final allocation is EF1.
Since agent $3q+1$ does not value any good, she is EF1 towards every other agent.
Therefore, we only need to consider agent $i$'s envy for $i \in \{1, \ldots, 3q\}$.
Note that there exists $j \in \{1, \ldots, p\}$ such that $x_i \in Y_j$ and $Y_j \in D$.
This means that $h_j$ is moved to some non-special agent's bundle in an exchange (possibly agent~$i$).
Regardless of whom $h_j$ goes to, agent $i$'s utility for her own bundle is at least $n_i - 2$, and agent $i$'s utility for agent $(3q+1)$'s bundle is exactly $n_i - 1$, so agent $i$ is EF1 towards agent $3q+1$.
Furthermore, if $h_j$ goes to some agent $i' \neq i$, then agent $i$'s utility for the bundle of agent $i'$ is $1$, so agent $i$ is EF1 towards agent $i'$.
Every other non-special agent's bundle yields zero utility to agent~$i$.
This shows that agent $i$ is EF1 towards every other agent.
Accordingly, the final allocation is EF1.

$(\Rightarrow)$ Suppose that after at most $q$ exchanges, an EF1 allocation is reached.
Let $i \in \{1, \ldots, 3q\}$.
The valuable goods from agent $i$'s perspective are with agent $i$ or with agent $3q+1$.
Since agent $i$'s utility for her own bundle in $\mathcal{A}$ is $n_i - 2$ and her utility for agent $(3q+1)$'s bundle is $n_i$, some valuable good from agent $(3q+1)$'s bundle needs to be moved to another agent's bundle (possibly $i$'s) in an exchange.
Now, each good in agent $(3q+1)$'s bundle is valuable to exactly three agents.
Since the movement of each good in agent $(3q+1)$'s bundle can only resolve the envy for at most three agents, at least $q$ goods need to be moved to make agents $1$ to $3q$ EF1.
This means that exactly $q$ exchanges are made; moreover, each good $h_j$ moved from agent $(3q+1)$'s bundle is associated with three distinct agents.
The set of $q$ goods moved from agent $(3q+1)$'s bundle induces an exact cover $D$ with cardinality $q$.
\end{proof}

Finally, we consider identical binary utilities.
We show that for this class of utilities, the computational problem can be solved efficiently regardless of whether the size vector of the initial allocation is balanced or not.
To show this, we demonstrate that the greedy algorithm allows an EF1 allocation to be reached using the smallest number of exchanges.

\begin{restatable}{theorem}{thmoptimalidenbingen} \label{thm:optimal_idenbin_gen}
\textsc{Optimal Exchanges} is in P for identical binary utilities.
\end{restatable}

\begin{proof}
Let $\mathcal{A} = (A_1, \ldots, A_n)$ be the given allocation.
As mentioned at the beginning of \Cref{sec:optimal}, we assume that an EF1 allocation can be reached from $\mathcal{A}$.
By the proof of \Cref{thm:exist_idenbin_gen}, there must be at most $s_0n + n - n_0$ valuable goods, where $s_0 = \min_{i\in N} |A_i|$ and $n_0 = |\{i \in N \mid |A_i| = s_0\}|$.
Suppose that there are $m_1 \leq s_0n + n - n_0$ valuable goods. 
An EF1 allocation requires every agent to receive at least $F := \lfloor m_1/n \rfloor$ valuable goods and at most $F+1$ valuable goods. 
Since an EF1 allocation with the same size vector as $\mathcal{A}$ exists, every agent must have at least $F$ goods in $\mathcal{A}$, i.e., $s_0 \geq F$.
Let $N_0$ be the set of agents who have at most $F$ valuable goods in the initial allocation, and $N_1$ be the set of agents who have at least $F+1$ valuable goods in the initial allocation. 
Note that $N = N_0 \cup N_1$. 
Let $c_0 = \sum_{i \in N_0} (F - u(A_i))$ and $c_1 = \sum_{i \in N_1} (u(A_i) - (F+1))$.
We claim that the optimal number of exchanges required to reach an EF1 allocation is $\max \{c_0, c_1\}$.
Note that this value can be computed in polynomial time, so it suffices to prove the claim.

First, we show that the optimal number of exchanges required to reach an EF1 allocation is at least $\max \{c_0, c_1\}$. 
Each agent $i \in N_0$ needs to receive at least $F - u(A_i) \geq 0$ valuable goods in order to arrive at a bundle with utility at least $F$. 
In receiving these valuable goods, agent $i$ must give away the same number of non-valuable goods from her bundle in $A_i$---note that this is possible since every agent has at least $F$ goods. 
Therefore, there exist at least $F - u(A_i)$ valuable goods from other agents' bundles that should go to agent $i$'s bundle and at least $F - u(A_i)$ non-valuable goods from agent $i$'s bundle that should go to other agents' bundles. 
Summing up over all $i \in N_0$, we have that at least $\sum_{i \in N_0} 2(F - u(A_i)) = 2c_0$ goods are in the wrong hands. 
Since each exchange places at most two goods in correct hands, the number of exchanges required is at least $2c_0/2 = c_0$.
By an analogous argument on the agents in $N_1$, we have that the number of exchanges required is at least $c_1$.
This proves that the optimal number of exchanges required to reach an EF1 allocation is at least $\max \{c_0, c_1\}$.

Next, we describe an algorithm that allows us to reach an EF1 allocation with at most $\max \{c_0, c_1\}$ exchanges. 
The algorithm is as follows: repeatedly exchange a valuable good from an agent with the highest utility with a non-valuable good from an agent with the lowest utility, until every agent has at least $F$ valuable goods and at most $F+1$ valuable goods.
We show that this ending will always be reached.
Suppose on the contrary that this is not the case, and consider the final allocation just before the algorithm cannot proceed further.
Since every agent has at least $F$ goods in $\mathcal{A}$, it must be possible that every agent receives at least $F$ valuable goods in the final allocation, and so $F = s_0$.
This means that some agent has more than $F+1$ valuable goods in the final allocation, and every agent who has $F$ goods in $\mathcal{A}$ has $F$ valuable goods in the final allocation.
Then, the number of valuable goods is $m_1 > Fn_0 + (F+1)(n-n_0) = s_0n + n - n_0$, which is a contradiction.
This shows that it is possible to reach the desired ending.

Now, we are ready to show that the optimal number of exchanges required is at most $\max \{c_0, c_1\}$.
If $c_0 \geq c_1$, then the first $c_1$ exchanges involve exchanging valuable goods from agents in $N_1$ with non-valuable goods from agents in $N_0$.
At this point, every agent in $N_1$ has exactly $F+1$ valuable goods, and every agent in $N_0$ has at most $F$ valuable goods.
Call this allocation $(B_1, \ldots, B_n)$.
We have that $\sum_{i \in N_0} (F - u(B_i)) = c_0 - c_1$.
If $|N_1| < c_0 - c_1$, then after $|N_1|$ further exchanges, every agent has at most $F$ valuable goods and some agent has fewer than $F$ valuable goods, contradicting the assumption that $F = \lfloor m_1/n \rfloor$.
Therefore, we must have $|N_1| \geq c_0 - c_1$.
Now, after $c_0 - c_1$ further exchanges, every agent in $N_0$ has exactly $F$ valuable goods and every agent in $N_1$ has between $F$ and $F+1$ valuable goods (inclusive), giving an EF1 allocation.
Hence, if $c_0 \geq c_1$, then the optimal number of exchanges required to reach an EF1 allocation is at most $c_1 + (c_0 - c_1) = c_0$.
By an analogous argument, if $c_0 < c_1$, then the optimal number of exchanges required to reach an EF1 allocation is at most $c_1$.
It follows that the optimal number of exchanges required to reach an EF1 allocation is at most $\max \{c_0, c_1\}$.
\end{proof}

\section{Worst-Case Bounds} \label{sec:worst}

In this section, instead of instance-specific optimization, we turn our attention to the \emph{worst-case} number of exchanges required to reach an EF1 allocation from some initial allocation.
Since an EF1 allocation may not always be reachable (as can be seen from \Cref{sec:exist}), we shall focus on the special case where the number of goods in each agent's bundle is the same.
We say that a size vector $\vec{s} = (s_1, \ldots, s_n)$ is \emph{$s$-balanced} for a positive integer $s$ if $s_i = s$ for all $i \in N$, and an allocation is \emph{$s$-balanced} if it has an $s$-balanced size vector.
We shall consider the worst-case number of exchanges starting from an $s$-balanced allocation for $n$ agents.

\subsection{General Utilities} \label{sec:worst_general}

Given $n$ and $s$, let $f(n, s)$ be the smallest integer such that for every instance with $n$ agents and $ns$ goods and every $s$-balanced allocation $\mathcal{A}$ in the instance, there exists an EF1 allocation that can be reached from~$\mathcal{A}$ using at most $f(n, s)$ exchanges.
We shall examine the bounds for $f(n, s)$.

We first derive an upper bound for $f(n, s)$.
At a high level, we use an algorithm by \citet{BiswasBa18} to find an EF1 allocation under cardinality constraints such that every agent retains roughly $s/n$ of her goods from her original bundle.
The algorithm also distributes the goods in each agent's initial bundle to the other agents as evenly as possible in order to maximize the number of goods that can be exchanged one-to-one, thereby minimizing the total number of exchanges.
One can check that roughly $s(n-1)/2$ exchanges are required to reach this EF1 allocation from the initial allocation.

\begin{restatable}{theorem}{thmworstub} \label{thm:worst_ub}
Let $n$ and $s$ be positive integers, and let $q = \lfloor s/n \rfloor$ and $r = s - qn$ be the quotient and remainder when $s$ is divided by $n$ respectively. Then,
\[
f(n, s) \leq
\begin{cases}
    s(n-1)/2                 & \text{if $r = 0$}; \\
    s(n-1)/2 + r(n-3)/2 + 1  & \text{otherwise}.
\end{cases}
\]
Moreover, we have $f(2, s) \leq (s-r)/2$ for all $s$.
\end{restatable}

\begin{proof}
Let $\mathcal{A}$ be an $s$-balanced allocation.
It suffices to find an EF1 $s$-balanced allocation~$\mathcal{B}$ such that the optimal number of exchanges to reach $\mathcal{B}$ from $\mathcal{A}$ is at most the expression given in the theorem statement.

When $n = 2$, allocate the goods in $A_1$ to the two agents in a round-robin fashion with agent $1$ going first, and allocate the goods in $A_2$ to the two agents in a round-robin fashion with agent $2$ going first.
Call this new allocation $\mathcal{B}$.
Note that $\mathcal{B}$ is clearly $s$-balanced.
We have $A_i \cap B_{3-i} = (s-r)/2$ for $i \in \{1, 2\}$, so the optimal number of exchanges required to reach $\mathcal{B}$ from $\mathcal{A}$ is (exactly) $(s-r)/2$.
To see that $\mathcal{B}$ is EF1, observe that agent $1$ does not envy agent $2$ with respect to the goods chosen from $A_1$ and is EF1 towards agent $2$ with respect to the goods chosen from $A_2$, so agent $1$ is EF1 towards agent~$2$ in $\mathcal{B}$; likewise, agent~$2$ is EF1 towards agent~$1$ in $\mathcal{B}$.
This shows that $f(2, s) \leq  (s-r)/2$.

When $n \geq 3$, we shall find an EF1 $s$-balanced allocation $\mathcal{B}$ by generalizing the method for two agents.
We define $n+r$ categories of goods $C_1, \ldots, C_n, D_1, \ldots, D_r$ as follows.
For $i \in N$, category~$C_i$ contains $qn$ goods arbitrarily selected from $A_i$ only; note that $r$ goods remain unselected in $A_i$.
Next, we form $D_w$ recursively as follows: let $w \in \{1, \ldots, r\}$ be the smallest index such that $D_w$ does not have $n$ goods yet, let $i \in N$ be the smallest index such that $A_i$ still has unselected goods, arbitrarily select a good in $A_i$, and add it to $D_w$.
At the end of this process, every category $C_i$ has exactly $qn$ goods from $A_i$, and every category $D_w$ has exactly $n$ goods from consecutive agents' bundles, say, $A_{i_w}, A_{i_w + 1}, \ldots, A_{j_w}$.

We now proceed to form $\mathcal{B}$ using the algorithm by \citet{BiswasBa18} which finds an EF1 allocation under cardinality constraints.
In particular, there exists an EF1 allocation $\mathcal{B} = (B_1, \ldots, B_n)$ such that $|C_i \cap B_j| = |C_i|/n = q$ for all $i, j \in N$ and $|D_w \cap B_j| = |D_w|/n = 1$ for all $w \in \{1, \ldots, r\},\, j \in N$.
Also, $\mathcal{B}$ is $s$-balanced because $|B_j| = qn + r = s$ for all $j \in N$.
We shall bound the number of exchanges required to reach $\mathcal{B}$ from $\mathcal{A}$.

For each unordered pair of distinct $i, j \in N$, exchange the $q$ goods from $C_i \cap B_j$ (which are in $A_i$) with the $q$ goods from $C_j \cap B_i$ (which are in $A_j$).
This requires a total of $qn(n-1)/2$ exchanges.
Call this intermediate allocation $\mathcal{A}' = (A'_1, \ldots, A'_n)$.
At this point, the only goods that are possibly in the wrong bundles in $\mathcal{A}'$ (as compared to $\mathcal{B}$) are the goods in all the $D_w$, and there are at most $rn$ such goods.
For each $i \in N$, let $X_i = A'_i \cap (D_1 \cup \cdots \cup D_r)$.

If $r = 0$, then $\mathcal{A}' = \mathcal{B}$, and we are done since the total number of exchanges is $qn(n-1)/2 = s(n-1)/2$.
Else, $r > 0$.
Consider the directed graph where the vertices are the agents and each edge~$e_g$ represents a good $g \in M$ such that if $g \in A'_i \cap B_j$, then $e_g = (i, j)$.
\citet[Prop.~4.1]{IgarashiKaSu24} showed that the number of exchanges required to reach $\mathcal{B}$ from $\mathcal{A}'$ is $m - c^*$, where $c^*$ is the maximum possible cardinality of a partition of the edges of the graph into (directed) circuits.
In $\mathcal{A}'$, $qn^2$ goods from all the $C_i$ are in the correct bundle by the previous process, and each of the edges representing these goods has its own circuit, say, $(i, i)$ if the good is in $A'_i$.
We shall show that the edges representing the $rn$ goods in all the $D_w$ can be partitioned into at least $2r-1$ disjoint circuits.
This will give at least $qn^2 + (2r - 1) = sn - (rn - 2r + 1)$ as the cardinality of one such partition of the edges of the graph into circuits.
Accordingly, $c^* \geq sn - (rn - 2r + 1)$, and the number of exchanges required to reach $\mathcal{B}$ from $\mathcal{A}'$ is $m - c^* \leq rn - 2r + 1$.
Then, the number of exchanges required to reach $\mathcal{B}$ from $\mathcal{A}$ (via $\mathcal{A}'$) is at most $$qn(n-1)/2 + (rn - 2r + 1) = s(n-1)/2 + r(n-3)/2 + 1,$$ establishing the theorem.

Let $w \in \{1, \ldots, r\}$ be given.
We shall show that there exists a cycle formed with a subset of the edges representing the goods in $D_w$.
The goods in $D_w$ come from consecutive agents' bundles in $\mathcal{A}'$, say, agents $i_w$ to $j_w$.
Every agent receives exactly one good from $D_w$ in $\mathcal{B}$; in particular, agents $i_w$ to $j_w$ receive exactly one good from $D_w$ each.
Consider the good~$g$ in $D_w \cap B_{i_w}$.
If $g$ is in $X_{i_w}$, then the edge $e_g = (i_w, i_w)$ is a desired cycle.
Otherwise, $g$ belongs to some agent $i' \in \{i_w + 1, \ldots, j_w\}$ in $\mathcal{A}'$.
Then, the edge $e_g$ is $(i', i_w)$.
Next, we consider the good $g'$ in $D_w \cap B_{i'}$, and find the agent that has $g'$ in $\mathcal{A}'$.
The edge representing~$g'$ then points to $i'$ from that agent.
By repeating this, we eventually find a cycle formed with some of these edges and with a subset of the agents $i_w$ to $j_w$ as vertices.
Let $M_w \subseteq D_w$ be the set of goods that are represented by the edges in this cycle.
Note that each $X_i$ for $i \in \{i_w, \ldots, j_w\}$ contains at most one good in $M_w$, and each $X_i$ for $i \in N \setminus \{i_w, \ldots, j_w\}$ does not contain any good in~$M_w$.

Now, consider the goods represented by the edges of the $r$ cycles---one for each $w$.
Note that these cycles are disjoint since the sets $M_w$ are pairwise disjoint.
Let $M_0 = \bigcup_{w=1}^r M_w$.
We claim that $|M_0| < 2n$.
Since the $r$ goods in $X_1$ are entirely contained in $D_1$, we have $|X_1 \cap M_1| \leq 1$ and $|X_1 \cap M_w| = 0$ for $w \in \{2, \ldots, r\}$, which implies that $|\bigcup_{w=1}^r (X_1 \cap M_w)| \leq 1$.
Now, for each $i \in N \setminus \{1\}$, the $r$ goods in $X_i$ can only be contained in at most two $D_w$---to see this, observe that if the $r$ goods are contained in $D_{w'}$, $D_{w'+1}$, and $D_{w'+2}$, then $D_{w'+1} \subseteq X_i$, which implies that $r = |X_i| \geq |D_{w'+1}| = n$, a contradiction.
Thus, we have $|X_i \cap M_w| \leq 1$ for all $w \in \{1, \ldots, r\}$, and $|X_i \cap M_w| = 1$ for at most two $w$, and so $|\bigcup_{w=1}^r (X_i \cap M_w)| \leq 2$.
Since $M_0 = \bigcup_{i \in N} \bigcup_{w=1}^r (X_i \cap M_w)$, we have $|M_0| \leq 1 + (n-1) \cdot 2 < 2n$, proving the claim.

Finally, consider the edges representing the $rn$ goods in all the $D_w$.
We have shown that fewer than $2n$ of these edges can be used to form $r$ disjoint circuits (in fact, cycles).
There are more than $rn - 2n = (r-2)n$ edges remaining.
Since we can always require every circuit to have length at most~$n$, there exists a partition of the remaining edges into more than $r-2$ disjoint circuits, i.e., at least $r-1$ disjoint circuits.
The total number of circuits in this partition is at least $r + (r-1) = 2r-1$.
This completes the proof.
\end{proof}

If no good is involved in more than one exchange, then $s(n-1)/2$ exchanges means that a total of $s(n-1) = m(1-1/n)$ goods are exchanged.
When $n$ is large, the fraction of goods involved in the exchanges becomes close to $1$.
While this bound might not seem impressive, we show next that it is, in fact, already essentially tight.
Specifically, we establish a lower bound for $f(n, s)$ by constructing an instance (with binary utilities) and an $s$-balanced allocation $\mathcal{A}$ in the instance such that roughly $s(n-1)/2$ exchanges are necessary to reach an EF1 allocation from $\mathcal{A}$.

\begin{restatable}{theorem}{thmworstlbbin} \label{thm:worst_lb_bin}
Let $n$ and $s$ be positive integers, and let $q = \lfloor s/n \rfloor$ and $r = s - qn$ be the quotient and remainder when $s$ is divided by $n$ respectively. Then,
\[
f(n, s) \geq
\begin{cases}
    s(n-1)/2            & \text{if $r = 0$}; \\
    s(n-1)/2 - (n-r)/2  & \text{otherwise}.
\end{cases}
\]
\end{restatable}

\begin{proof}
Let $M = \{ g_{i, j} \mid 1 \leq i \leq n, 1 \leq j \leq s \}$ be the set of goods such that each good $g_{i, j}$ is worth~$0$ to agent $i$ and worth $1$ to all agents except $i$.
We have $u_i(M) = s(n-1)$.
We claim that an EF1 allocation requires every agent to receive a bundle worth at least $s - q - \lceil r/n \rceil$ from her perspective.
To see this, suppose on the contrary that some agent~$i$ receives a bundle worth less than $s - q - \lceil r/n \rceil$ to her.
For the allocation to be EF1, every other agent receives a bundle worth at most $s - q - \lceil r/n \rceil$ to agent $i$.
Then, we must have $u_i(M) < n(s - q - \lceil r/n \rceil)$.
When $r = 0$, it holds that $\lceil r/n \rceil = 0$ and $$n(s - q - \lceil r/n \rceil) = n(s - q) = sn - s = s(n-1).$$
When $r > 0$, it holds that $\lceil r/n \rceil = 1$ and $$n(s - q - \lceil r/n \rceil) = n(s - q - 1) = sn - (qn + r) - (n-r) = s(n-1) - (n-r) \leq s(n-1).$$
In both cases, we have $u_i(M) < s(n-1) = u_i(M)$, a contradiction.

Let $\mathcal{A}$ be the allocation such that $A_i = \{g_{i, j} \mid 1 \leq j \leq s\}$ for every $i$. 
In order to reach an EF1 allocation, each agent must give away at least $s - q - \lceil r/n \rceil$ goods from her bundle in order to receive from the other agents the same number of valuable goods from her perspective.
The total number of goods that are currently in the wrong hands across all agents is at least $n(s - q - \lceil r/n \rceil)$, and the optimal number of exchanges required to reach an EF1 allocation is at least half of this number, since each exchange places at most two goods in the correct hands.
When $r = 0$, the optimal number of exchanges required is at least $n(s - q - \lceil r/n \rceil)/2 = s(n-1)/2$.
When $r > 0$, the optimal number of exchanges required is at least $n(s - q - \lceil r/n \rceil)/2 = s(n-1)/2 - (n-r)/2$.
\end{proof}

For two agents, \Cref{thm:worst_ub,thm:worst_lb_bin} give a tight bound of $f(2, s) = (s-r)/2 = m/4 - r/2 = \lfloor m/4 \rfloor$.
This means that in the worst-case scenario, the number of exchanges required to reach an EF1 allocation is roughly one-quarter of the total number of goods between the two agents, or equivalently, roughly half of the goods need to be exchanged between the two agents to reach an EF1 allocation.

\Cref{thm:worst_ub,thm:worst_lb_bin} also give a tight bound of $f(n, s) = s(n-1)/2$ whenever $s$ is divisible by~$n$.
By observing the proof of \Cref{thm:worst_ub}, we can achieve an EF1 allocation with $f(n, s)$ exchanges without involving each good in more than one exchange.
This means that a $(1 - 1/n)$ fraction of all goods need to be exchanged in the worst-case scenario.
Intuitively, this happens when each agent only values the goods in the bundle of every agent except her own in the initial allocation, and therefore needs to ensure that these goods are evenly distributed among all agents including herself.

Define $\fbin(n, s)$ as the smallest integer such that for every \emph{binary} instance with $n$ agents and $ns$ goods and every $s$-balanced allocation $\mathcal{A}$ in the instance, there exists an EF1 allocation that can be reached from $\mathcal{A}$ using at most $\fbin(n, s)$ exchanges.
The proof of \Cref{thm:worst_lb_bin} uses a binary instance, which means that the lower bound of the theorem works for $\fbin$ as well.
Clearly, the upper bound of \Cref{thm:worst_ub} works for $\fbin$, so the discussion in the preceding paragraphs also applies to binary instances too.

\subsection{Identical Binary Utilities} \label{sec:worst_idenbin}

Given $n$ and $s$, let $\fidbin (n, s)$ be the smallest integer such that for every instance with $n$ agents with \emph{identical binary} utilities and $ns$ goods and every $s$-balanced allocation $\mathcal{A}$ in the instance, there exists an EF1 allocation that can be reached from $\mathcal{A}$ using at most $\fidbin (n, s)$ exchanges.
We show that $\fidbin (n, s)$ is roughly $sn/4$ for even $n$ and $s(n-1)(n+1)/4n$ for odd $n$---note that this is approximately half of the bound $f(n,s)$ for arbitrary utilities, which is roughly $s(n-1)/2$ as seen earlier.
The upper bounds (of $sn/4$ and $s(n-1)(n+1)/4n$ respectively) correspond to the case where half of the agents have all the valuable goods while the remaining half have all the non-valuable goods.

\begin{restatable}{theorem}{thmworstidenbin} \label{thm:worst_idenbin}
Let $n$ and $s$ be positive integers.
If $n$ is even, then
$\frac{n}{2}\left\lfloor \frac{s}{2} \right\rfloor \leq \fidbin (n, s) \leq \frac{sn}{4}$.
If $n$ is odd, then $\frac{n+1}{2}\left\lfloor \frac{s(n-1)}{2n} \right\rfloor \leq \fidbin (n, s) \leq \frac{s(n-1)(n+1)}{4n}$.
\end{restatable}

\begin{proof}
Recall that the proof of \Cref{thm:optimal_idenbin_gen} provides a way to compute the optimal number of exchanges to reach an EF1 allocation from a given initial allocation.
To recap, let $m_1$ be the total number of valuable goods, $F = \lfloor m_1 / n \rfloor$ be the minimum number of valuable goods each agent must receive in an EF1 allocation, $N_0$ be the set of agents who have at most $F$ valuable goods in the initial allocation, $N_1$ be the set of agents who have at least $F+1$ valuable goods in the initial allocation, $c_0 = \sum_{i \in N_0} (F - u(A_i))$, and $c_1 = \sum_{i \in N_1} (u(A_i) - (F+1))$.
The optimal number of exchanges is $\max \{c_0, c_1\}$.

We first prove the lower bounds for $\fidbin (n, s)$ by providing an explicit initial allocation and showing that the optimal number of exchanges to reach an EF1 allocation is at least $\left\lfloor \left\lfloor n/2 \right\rfloor s/n \right\rfloor \cdot \left\lceil n/2 \right\rceil$, which corresponds to the lower bounds for both even and odd $n$.
In the initial allocation, $\lfloor n/2 \rfloor$ agents have $s$ valuable goods each and the remaining $\lceil n/2 \rceil$ agents have $s$ non-valuable goods each.
There are a total of $m_1 = \lfloor n/2 \rfloor \cdot s$ valuable goods, and $F = \left\lfloor \left\lfloor n/2 \right\rfloor s/n \right\rfloor$.
The value of $c_0$ is $$\sum_{i \in N_0} (F - u(A_i)) = \sum_{i \in N_0} (\left\lfloor \left\lfloor n/2 \right\rfloor s/n \right\rfloor - 0) = \left\lfloor \left\lfloor n/2 \right\rfloor s/n \right\rfloor \cdot \left\lceil n/2 \right\rceil.$$
Since $\max \{c_0, c_1\} \geq c_0$, the lower bounds follow.

We now prove the upper bounds for $\fidbin (n, s)$.
Let an $s$-balanced allocation $\mathcal{A} = (A_1, \ldots, A_n)$ be given, and let $n_0 = |N_0|$ and $n_1 = |N_1|$.
We first derive upper bounds for $c_0$ and $c_1$.
Note that $m_1 \leq sn_1 + \sum_{i \in N_0} u(A_i)$.
We have
\begin{align*}
    c_0 = \sum_{i \in N_0} (F - u(A_i)) 
    &= (n - n_1)F - \sum_{i \in N_0} u(A_i) \\
    &\leq (n - n_1)\frac{m_1}{n} - \sum_{i \in N_0} u(A_i) \\
    &= \left(1 - \frac{n_1}{n}\right) m_1 - \sum_{i \in N_0} u(A_i) \\
    &\leq \left(1 - \frac{n_1}{n}\right) \left(sn_1 + \sum_{i \in N_0} u(A_i)\right) - \sum_{i \in N_0} u(A_i) \\
    &\leq sn_1 - \frac{sn_1^2}{n} + \sum_{i \in N_0} u(A_i) - \sum_{i \in N_0} u(A_i) 
    = \frac{sn_1}{n}(n - n_1).
\end{align*}
On the other hand, $m_1 \geq \sum_{i \in N_1} u(A_i)$.
We have
\begin{align*}
    c_1 = \sum_{i \in N_1} (u(A_i) - (F+1)) 
    &= \sum_{i \in N_1} u(A_i) - n_1\left( \left\lfloor \frac{m_1}{n} \right\rfloor +1\right) \\
    &\leq \sum_{i \in N_1} u(A_i) - n_1 \left(\frac{m_1}{n}\right) \\
    &\leq \sum_{i \in N_1} u(A_i) - \frac{n_1}{n} \sum_{i \in N_1} u(A_i) \\
    &= \frac{1}{n}(n - n_1) \sum_{i \in N_1} u(A_i) 
    \leq \frac{1}{n}(n - n_1) sn_1 
    = \frac{sn_1}{n}(n - n_1).
\end{align*}

We have thus shown that $\max \{c_0, c_1\} \leq n_1(n - n_1)s/n$.
Therefore, the optimal number of exchanges is at most $n_1(n - n_1)s/n$, which is a quadratic expression in $n_1$.
When $n$ is even, $n_1(n - n_1)s/n$ attains a maximum value at $n_1 = n/2$, and this value is $sn/4$.
When $n$ is odd, $n_1(n - n_1)s/n$ attains a maximum value at $n_1 = (n+1)/2$ and $n_1 = (n-1)/2$, and this value is $s(n-1)(n+1)/4n$.
The upper bounds for $\fidbin (n, s)$ follow.
\end{proof}

\section{Conclusion and Future Work}

In this paper, we have studied the reformability of unfair allocations and the number of exchanges required in the reformation process.
We uncovered several distinctions in the complexity of these problems based on the number of agents and their utility functions, and showed that the number of exchanges required to reach an EF1 allocation is relatively high in the worst case.

While our worst-case bounds for general utilities are already exactly tight in certain scenarios and almost tight generally, an open question is to tighten them for more than two agents when the number of goods in each agent's bundle is not divisible by the number of agents.
Additionally, although these bounds also work for binary utilities, one could try to derive bounds for \emph{identical} utilities.
We provide some insights for identical (but not necessarily binary) utilities in Appendix~\ref{ap:worst_iden}.
Another interesting direction is to require each exchange to be beneficial for both agents involved---in Appendix~\ref{ap:beneficial}, we prove that the problem of deciding whether a given initial allocation can be reformed into an EF1 allocation using only beneficial exchanges is NP-complete for binary utilities.
One could also consider the model of \emph{transferring} goods instead of exchanging them; an EF1 allocation is always reachable from any allocation in this model, so a natural question is to determine the optimal number of exchanges needed for this goal.
In Appendix~\ref{ap:transfer}, we show that several of our proof ideas for exchanges can be adapted to handle transfers as well.
Finally, one could consider reforming an allocation using notions other than EF1 as fairness benchmarks, or allow utilities that are not necessarily additive.\footnote{An EF1 allocation always exists even for arbitrary monotonic utilities \citep{LiptonMaMo04}.
For such utilities, given that a single utility function may already take exponential time to describe, it is common to assume a \emph{query model}, where an algorithm can query an agent's utility for a set of goods.
Interestingly, we observe that even for two agents with identical utilities, an exponential number of queries may be necessary to answer \textsc{Reformability} for monotonic utilities (cf.~\Cref{thm:exist_iden_two}).
To see this, let $m = 2k+2$ and let the size vector be $(k, k+2)$.
Suppose that the utility of each set of goods is equal to its size, except for a particular set $A$ of size $k$, whose utility may be $k$ or $k+1$ (the set~$A$ is not known to the algorithm).
The answer to \textsc{Reformability} is ``Yes'' exactly when $u(A) = k+1$, but any algorithm needs an exponential number of queries in the worst case to discover this.
}

\section*{Acknowledgments}

This work was partially supported by the Singapore Ministry of Education under grant
number MOE-T2EP20221-0001, by JST FOREST under grant
number JPMJFR226O, by JSPS KAKENHI under grant number JP20H05795, by JST ERATO under grant number JPMJER2301, and by an NUS Start-up Grant.
We would like to thank the ISAAC 2025 reviewers and participants for their valuable comments.

\bibliographystyle{plainnat}
\bibliography{main}

\appendix

\section{Worst-Case Bounds for Identical Utilities}
\label{ap:worst_iden}

We continue the discussion from \Cref{sec:worst} on worst-case bounds, and focus on identical utilities in this appendix.

Given $n$ and $s$, let $\fid(n, s)$ be the smallest integer such that for every instance with $n$ agents with \emph{identical} utilities and $ns$ goods and every $s$-balanced allocation $\mathcal{A}$ in the instance, there exists an EF1 allocation that can be reached from $\mathcal{A}$ using at most $\fid(n, s)$ exchanges.

A tight bound for two agents is an immediate consequence of our previous results.
Indeed, the lower bound follows from \Cref{thm:worst_idenbin}, while the upper bound follows from \Cref{thm:worst_ub}.

\begin{theorem}
Let $s$ be a positive integer.
Then, $\fid(2, s) = \lfloor s/2 \rfloor$.
\end{theorem}

For three or more agents, we conjecture that $\fid(n, s)$ is roughly $sn/4$, like $\fidbin(n, s)$.
However, proving this turns out to be surprisingly challenging.
We shall present a result using a slightly weaker fairness notion in the case of three agents.

We say that agent $i$ is \emph{weak-EF1 towards} agent $j$ in an allocation $\mathcal{A} = (A_1, \ldots, A_n)$ if $u_i(A_i) \geq u_i(A_j) - \max_{g \in M} u_i(g)$; note the condition $g\in M$ as opposed to $g\in A_j$ for EF1.
An allocation $\mathcal{A}$ is \emph{weak-EF1} if every agent is weak-EF1 towards every other agent in $\mathcal{A}$.
Weak-EF1 is the fairness notion originally considered by \citet{LiptonMaMo04} (although their algorithm satisfies EF1), and weak-EF1 and EF1 are equivalent when the utilities are binary.
Since we consider identical utilities, we use $u$ instead of $u_i$.
Without loss of generality, we may divide all utilities by $\max_{g \in M} u(g)$.
Then, the utility of each good is in $[0,1]$, and the condition for agent $i$ to be weak-EF1 towards agent~$j$ is $u(A_i) \geq u(A_j) - 1$.

Given $n$ and $s$, let $\fidweak(n, s)$ be the smallest integer such that for every instance with $n$ agents with identical utilities and $ns$ goods, and every $s$-balanced allocation $\mathcal{A}$ in the instance, there exists a weak-EF1 allocation that can be reached from $\mathcal{A}$ using at most $\fidweak(n, s)$ exchanges.
We shall determine the value of $\fidweak(3, s)$.

We describe an algorithm $\mathfrak{A}$ that performs a sequence of exchanges of goods starting from an initial allocation $\mathcal{A}^0$.
For each $t$ starting from $0$, we begin with the allocation $\mathcal{A}^t = (A_1^t, \ldots, A_n^t)$.
If $\mathcal{A}^t$ is weak-EF1, then we are done and the algorithm terminates.
Otherwise, we perform an exchange of goods between two agents to reach the allocation $\mathcal{A}^{t+1} = (A_1^{t+1}, \ldots, A_n^{t+1})$.
For each agent $k$, let $g_k^t$ and $h_k^t$ be a good of the highest utility and a good of the lowest utility in agent $k$'s bundle, $A_k^t$, respectively.
Let $i_t$ be an agent with the most valuable bundle, i.e., $i_t = \argmax_{k\in N} u(A_k^t)$, and $j_t$ be an agent with the least valuable bundle, i.e., $j_t = \argmin_{k\in N} u(A_k^t)$; we may resolve ties arbitrarily.
Note that agent $j_t$ is not weak-EF1 towards agent $i_t$---otherwise, $\mathcal{A}^t$ is weak-EF1---and hence $i_t \neq j_t$.
We then exchange $g_{i_t}^t$ with $h_{j_t}^t$ to form $\mathcal{A}^{t+1}$, i.e., $A_{i_t}^{t+1} = (A_{i_t}^t \setminus \{g_{i_t}^t\}) \cup \{h_{j_t}^t\}$, $A_{j_t}^{t+1} = (A_{j_t}^t \setminus \{h_{j_t}^t\}) \cup \{g_{i_t}^t\}$, and $A_k^{t+1} = A_k^t$ for all $k \in N \setminus \{i_t, j_t\}$.
Subsequently, we increment $t$ by $1$ and repeat the procedure.

To establish our result, we prove a series of lemmas on properties of this algorithm.

\begin{lemma}
\label{lem:goods_comparison}
Let $\mathcal{A}^t$ be an allocation which is not weak-EF1.
Then, $u(g_{i_t}^t) > u(h_{j_t}^t)$.
\end{lemma}

\begin{proof}
If $u(g_{i_t}^t) \leq u(h_{j_t}^t)$, then
\begin{align*}
    u(A_{i_t}^t) - 1 &\leq u(A_{i_t}^t) \leq s \cdot u(g_{i_t}^t) \leq s \cdot u(h_{j_t}^t) \leq u(A_{j_t}^t),
\end{align*}
so agent $j_t$ is weak-EF1 towards agent $i_t$, and therefore $\mathcal{A}^t$ is weak-EF1, a contradiction.
Hence, $u(g_{i_t}^t) > u(h_{j_t}^t)$.
\end{proof}

\begin{lemma}
\label{lem:agents_ef1}
Let $\mathcal{A}^t$ be an allocation which is not weak-EF1.
Then, in $\mathcal{A}^{t+1}$,
\begin{itemize}
    \item agent $i_t$ is weak-EF1 towards every agent; and
    \item every agent is weak-EF1 towards agent $j_t$.
\end{itemize}
\end{lemma}

\begin{proof}
Let $k \in N \setminus \{i_t, j_t\}$. Note that $u(A_k^{t+1}) = u(A_k^t)$.

Since $i_t = \argmax_{\ell\in N}u(A_\ell^t)$, we have
\begin{align*}
    u(A_{i_t}^{t+1}) &= u((A_{i_t}^t \setminus \{g_{i_t}^t\}) \cup \{h_{j_t}^t\}) 
    \geq u(A_{i_t}^t \setminus \{g_{i_t}^t\}) 
    \geq u(A_{i_t}^t) - 1 
    \geq u(A_k^t) - 1 
    = u(A_k^{t+1}) - 1,
\end{align*}
showing that agent $i_t$ is weak-EF1 towards agent $k$.

Similarly, since $j_t = \argmin_{\ell\in N} u(A_\ell^t)$, we have
\begin{align*}
    u(A_k^{t+1}) &= u(A_k^t) 
    \geq u(A_{j_t}^t) 
    = u((A_{j_t}^{t+1} \cup \{h_{j_t}^t\}) \setminus \{g_{i_t}^t\}) 
    \geq u(A_{j_t}^{t+1} \setminus \{g_{i_t}^t\}) 
    \geq u(A_{j_t}^{t+1}) - 1,
\end{align*}
showing that agent $k$ is weak-EF1 towards agent $j_t$.

Finally, since agent $j_t$ is not weak-EF1 towards agent $i_t$, we have $u(A_{j_t}^t) < u(A_{i_t}^t) - 1$. Thus,
\begin{align*}
    u(A_{i_t}^{t+1}) &= u((A_{i_t}^t \setminus \{g_{i_t}^t\}) \cup \{h_{j_t}^t\}) \\
    &\geq u(A_{i_t}^t \setminus \{g_{i_t}^t\}) \\
    &\geq u(A_{i_t}^t) - 1 
    > u(A_{j_t}^t) 
    = u((A_{j_t}^{t+1} \cup \{h_{j_t}^t\}) \setminus \{g_{i_t}^t\}) 
    \geq u(A_{j_t}^{t+1} \setminus \{g_{i_t}^t\}) 
    \geq u(A_{j_t}^{t+1}) - 1,
\end{align*}
showing that agent $i_t$ is weak-EF1 towards agent $j_t$.
\end{proof}

For each $t \geq 0$, call $i_t$ a \emph{strong agent} and $j_t$ a \emph{weak agent}.
Let $I^0 = J^0 = \emptyset$, and for each $t \geq 0$, let $I^{t+1} = I^t \cup \{i_t\}$ be the set of strong agents up to round $t$, and $J^{t+1} = J^t \cup \{j_t\}$ be the set of weak agents up to round $t$.

\begin{lemma}
\label{lem:strong_weak_disjoint}
Let $t \geq 0$ be given such that $\mathcal{A}^0, \ldots, \mathcal{A}^t$ are not weak-EF1.
Then, $I^{t+1} \cap J^{t+1} = \emptyset$.
\end{lemma}

\begin{proof}
Suppose on the contrary that there exists an agent $k$ such that $k \in I^{t+1} \cap J^{t+1}$.
Let $t_p$ be the smallest index such that $k \in I^{t_p+1}$, and $t_q$ be the smallest index such that $k \in J^{t_q+1}$.
Then, we have $k = i_{t_p} = j_{t_q}$.
Note that $t_p \neq t_q$, since $i_{t'} \neq j_{t'}$ for all $t'$.

Suppose first that $t_p < t_q$.
We show by induction that agent $k$ is weak-EF1 towards every agent in $\mathcal{A}^{t'+1}$ for all $t' \in \{t_p, \ldots, t\}$.
The base case of $t' = t_p$ is true by \Cref{lem:agents_ef1} since $k = i_{t_p}$.
For the inductive step, suppose that agent $k$ is weak-EF1 towards every agent in $\mathcal{A}^{t'+1}$ for some $t' \in \{t_p, \ldots, t-1\}$.
Then, agent $k$ cannot be $j_{t'+1}$.
If agent $k$ is $i_{t'+1}$, then agent $k$ is weak-EF1 towards every agent in $\mathcal{A}^{t'+2}$ by \Cref{lem:agents_ef1}, making the inductive statement true.
If agent $k$ is not $i_{t'+1}$, then agent $k$ does not take part in the exchange going from $\mathcal{A}^{t'+1}$ to $\mathcal{A}^{t'+2}$.
\begin{itemize}
    \item Agent $k$ is weak-EF1 towards agent $i_{t'+1}$ in $\mathcal{A}^{t'+2}$ since agent $k$ is weak-EF1 towards $i_{t'+1}$ in~$\mathcal{A}^{t'+1}$ by the inductive hypothesis, and agent $i_{t'+1}$'s utility of her own bundle decreases after the exchange by \Cref{lem:goods_comparison}.
    \item Agent $k$ is weak-EF1 towards agent $j_{t'+1}$ in $\mathcal{A}^{t'+2}$ by \Cref{lem:agents_ef1}.
    \item Agent $k$ is weak-EF1 towards every other agent in $\mathcal{A}^{t'+2}$ since their bundles did not change from~$\mathcal{A}^{t'+1}$.
\end{itemize}
Overall, these show that agent $k$ is weak-EF1 towards every agent in $\mathcal{A}^{t'+2}$, proving the inductive statement.
Since agent~$k$ is weak-EF1 towards every agent in $\mathcal{A}^{t'+1}$ for all $t' \in \{t_p, \ldots, t\}$, agent $k$ can never be $j_{t_q}$.
This shows that $t_p < t_q$ is false.

Therefore, we must have $t_p > t_q$.
The argument for this case is similar to that for the previous case.
We show by induction that every agent is weak-EF1 towards agent $k$ in $\mathcal{A}^{t'+1}$ for all $t' \in \{t_q, \ldots, t\}$.
The base case of $t' = t_q$ is true by \Cref{lem:agents_ef1} since $k = j_{t_q}$.
For the inductive step, suppose that every agent is weak-EF1 towards agent $k$ in $\mathcal{A}^{t'+1}$ for some $t' \in \{t_q, \ldots, t-1\}$.
Then, agent $k$ cannot be~$i_{t'+1}$.
If agent $k$ is $j_{t'+1}$, then every agent is weak-EF1 towards agent~$k$ in $\mathcal{A}^{t'+2}$ by \Cref{lem:agents_ef1}, making the inductive statement true.
If agent $k$ is not $j_{t'+1}$, then agent $k$ does not take part in the exchange going from $\mathcal{A}^{t'+1}$ to $\mathcal{A}^{t'+2}$.
\begin{itemize}
    \item Agent $i_{t'+1}$ is weak-EF1 towards agent $k$ in $\mathcal{A}^{t'+2}$ by \Cref{lem:agents_ef1}.
    \item Agent $j_{t'+1}$ is weak-EF1 towards agent $k$ in $\mathcal{A}^{t'+2}$ since agent $j_{t'+1}$ is weak-EF1 towards agent~$k$ in $\mathcal{A}^{t'+1}$ by the inductive hypothesis, and agent $j_{t'+1}$'s utility of her own bundle increases after the exchange by \Cref{lem:goods_comparison}.
    \item Every other agent is weak-EF1 towards agent $k$ in $\mathcal{A}^{t'+2}$ since their bundles did not change from~$\mathcal{A}^{t'+1}$.
\end{itemize}
Overall, these show that every agent is weak-EF1 towards agent $k$ in $\mathcal{A}^{t'+2}$, proving the inductive statement.
Since every agent is weak-EF1 towards agent $k$ in $\mathcal{A}^{t'+1}$ for all $t' \in \{t_q, \ldots, t\}$, agent $k$ can never be $i_{t_p}$.
This yields the desired contradiction.
\end{proof}

\begin{lemma}
\label{lem:monotone}
Let $t \geq 0$ be given such that $\mathcal{A}^0, \ldots, \mathcal{A}^t$ are not weak-EF1.
Then,
\begin{itemize}
    \item for any $i \in I^{t+1}$, $u(A_i^0) \geq \cdots \geq u(A_i^{t+1})$; and
    \item for any $j \in J^{t+1}$, $u(A_j^0) \leq \cdots \leq u(A_j^{t+1})$.
\end{itemize}
\end{lemma}

\begin{proof}
At every time step $t' \in \{0, \ldots, t\}$, an agent $i \in I^{t+1}$ cannot be a weak agent by \Cref{lem:strong_weak_disjoint}.
Therefore, agent~$i$ either takes part in the exchange from $\mathcal{A}^{t'}$ to $\mathcal{A}^{t'+1}$ as a strong agent $i_{t'}$ or does not take part in the exchange.
The utility of agent $i$'s bundle either decreases in the former case due to \Cref{lem:goods_comparison} or remains the same in the latter case.
An analogous argument holds for $j \in J^{t+1}$.
\end{proof}

\begin{lemma}
\label{lem:goods}
Each good is not exchanged more than once in algorithm $\mathfrak{A}$.
\end{lemma}

\begin{proof}
Suppose on the contrary that some good $g$ is exchanged more than once.
We first consider the case where $g$ is in a strong agent's bundle in $\mathcal{A}^0$ and is exchanged for the first time at round $t$, i.e., $g = g_{i_t}^t$.
After its first exchange, the good is now with agent $j_t$.
By \Cref{lem:strong_weak_disjoint}, $j_t \notin I^{t'}$ for any $t' > t$.
Since the good is exchanged again, it must be that $g = h_{j_{t'}}^{t'}$ for some $t' > t$, where $j_{t'} = j_t$.
Then, we have
\begin{align*}
    u(A_{i_{t'}}^t) - 1 
    &\geq u(A_{i_{t'}}^{t'}) - 1 \tag*{(by \Cref{lem:monotone} on $i_{t'} \in I^{t'+1}$)} \\
    &> u(A_{j_{t'}}^{t'}) \tag*{(since $j_{t'}$ is not weak-EF1 towards $i_{t'}$)}\\
    &\geq s \cdot u( g ) \tag*{(since $g$ is the least valuable good in $A_{j_{t'}}^{t'}$)} \\
    &\geq u(A_{i_t}^t) \tag*{(since $g$ is the most valuable good in $A_{i_t}^t$)} \\
    &> u(A_{i_t}^t) - 1,
\end{align*}
which means that agent $j_t$ should have exchanged goods with agent $i_{t'}$ at round $t$ instead of with agent~$i_t$.
This contradiction shows that a good in a strong agent's bundle in $\mathcal{A}^0$ cannot be exchanged more than once.

Analogously, we now consider the case where $g$ is in a weak agent's bundle in $\mathcal{A}^0$ and is exchanged for the first time at round $t$, i.e., $g = h_{j_t}^t$.
After its first exchange, the good is now with agent $i_t$.
By \Cref{lem:strong_weak_disjoint},  $i_t \notin J^{t'}$ for any $t' > t$.
Since the good is exchanged again, it must be that $g = g_{i_{t'}}^{t'}$ for some $t' > t$, where $i_{t'} = i_t$.
Then, we have
\begin{align*}
    u(A_{j_{t'}}^t) 
    &\leq u(A_{j_{t'}}^{t'}) \tag*{(by \Cref{lem:monotone} on $j_{t'} \in J^{t'+1}$)} \\
    &< u(A_{i_{t'}}^{t'}) - 1 \tag*{(since $j_{t'}$ is not weak-EF1 towards $i_{t'}$)} \\
    &< u(A_{i_{t'}}^{t'}) \\
    &\leq s \cdot u(g) \tag*{(since $g$ is the most valuable good in $A_{i_{t'}}^{t'}$)} \\
    &\leq u(A_{j_t}^t), \tag*{(since $g$ is the least valuable good in $A_{j_t}^t$)}
\end{align*}
which means that agent $i_t$ should have exchanged goods with agent $j_{t'}$ at round $t$ instead of with agent~$j_t$.
This contradiction shows that a good in a weak agent's bundle in $\mathcal{A}^0$ also cannot be exchanged more than once.
\end{proof}

\begin{lemma}
\label{lem:terminates}
Algorithm $\mathfrak{A}$ terminates in finite time.
\end{lemma}

\begin{proof}
Since each good is not exchanged more than once by \Cref{lem:goods}, at most $\lfloor m/2 \rfloor$ pairs of goods can be exchanged, and the algorithm terminates by round $\lfloor m/2 \rfloor$.
\end{proof}

Since the algorithm terminates in finite time by \Cref{lem:terminates}, there exists $T \geq 0$ such that $\mathcal{A}^0, \ldots, \mathcal{A}^T$ are not weak-EF1 but $\mathcal{A}^{T+1}$ is weak-EF1.
Let $I = I^{T+1}$ be the set of strong agents and $J = J^{T+1}$ be the set of weak agents.
By \Cref{lem:strong_weak_disjoint}, $I$ and $J$ are disjoint sets of agents.
Therefore, at each round $t \in \{0, \ldots, T\}$ of the algorithm, some agent $i_t \in I$ exchanges a good with some agent $j_t \in J$.

We derive a bound on the number of steps that $\mathfrak{A}$ takes in the case of two agents.

\begin{lemma}
\label{lem:worst_iden_two_agents}
For $n = 2$ agents with $s$ goods each, algorithm $\mathfrak{A}$ terminates after at most $\lfloor s/2 \rfloor$ rounds.
\end{lemma}

\begin{proof}
The statement is clear when $s = 1$, so we assume that $s \geq 2$.
Suppose on the contrary that after $T = \lfloor s/2 \rfloor$ rounds, the allocation $\mathcal{A}^T$ is still not weak-EF1.
Without loss of generality, assume that $1 \in I$ and $2 \in J$.
Then, the most valuable $\lfloor s/2 \rfloor$ goods from agent $1$'s bundle $A_1^0$ are exchanged with the least valuable $\lfloor s/2 \rfloor$ goods from agent~$2$'s bundle $A_2^0$ to reach $\mathcal{A}^T$.
Let $B_1 \subseteq A_1^0$ and $B_2 \subseteq A_2^0$ be the sets of goods from the respective bundles that are exchanged between the two agents, and let $C_1 = A_1^0 \setminus B_1$ and $C_2 = A_2^0 \setminus B_2$.
Note that all these sets are disjoint by \Cref{lem:goods}.
Let $g$ be an arbitrary good in $C_1$, and let $C'_1 = C_1 \setminus \{g\}$.
We have $|B_1| = |B_2| = \lfloor s/2 \rfloor$, $|C_1| = |C_2| = \lceil s/2 \rceil$, and $|C'_1| \leq \lfloor s/2 \rfloor$.

Now, $u(B_1) \geq u(C'_1)$ since the goods with the highest values from $A_1^0$ are exchanged and $B_1$ has at least as many goods as $C'_1$.
Also, $u(C_2) \geq u(B_2)$ since the goods with the lowest values from $A_2^0$ are exchanged and $C_2$ has at least as many goods as $B_2$.
Therefore, we have
\begin{align*}
    u(A_2^T) &= u(B_1 \cup C_2) 
    \geq u(C'_1 \cup B_2) 
    = u(A_1^T \setminus \{g\}) 
    \geq u(A_1^T) - 1,
\end{align*}
which shows that agent $2$ is weak-EF1 towards agent $1$ in $\mathcal{A}^T$.
On the other hand, agent $1$ is also weak-EF1 towards agent $2$ in $\mathcal{A}^T$ due to \Cref{lem:agents_ef1} applied on $\mathcal{A}^{T-1}$.
This shows that $\mathcal{A}^T$ is weak-EF1, a contradiction.
\end{proof}

We now come to our main lemma, which bounds the number of steps that $\mathfrak{A}$ takes for three agents.
For convenience of the analysis, we focus on the case where $s$ is divisible by $3$.

\begin{lemma}
\label{lem:worst_iden_three_agents}
Let $s$ be a positive integer divisible by $3$.
For $n = 3$ agents with $s$ goods each, algorithm~$\mathfrak{A}$ terminates after at most $2s/3$ rounds.
\end{lemma}

\begin{proof}
Suppose on the contrary that after $T = 2s/3$ rounds, the allocation $\mathcal{A}^T$ is still not weak-EF1.
Note that $T > 0$, so $I^T, J^T \neq \emptyset$.
If $|I^T| = |J^T| = 1$, then after at most $\lfloor s/2 \rfloor$ rounds, the agent $i \in I^T$ and the agent $j \in J^T$ are weak-EF1 towards each other by \Cref{lem:worst_iden_two_agents}, while the agent $k \in N \setminus \{i, j\}$ is weak-EF1 towards everyone and vice versa since agent $k$ does not partake in the exchanges.
Since $\lfloor s/2 \rfloor \leq 2s/3$, the allocation $\mathcal{A}^{T'}$ is weak-EF1 for some $T' \leq 2s/3$, contradicting our assumption.
Therefore, we must have $I^T \cup J^T = N$.

\vspace{2mm}

\textbf{Case 1: $|I^T| = 1$. }
We consider the allocation $\mathcal{A}^T$ relative to $\mathcal{A}^0$.
Without loss of generality, let $1 \in I^T$.
Let $B_{1,2} \subseteq A_1^0$ and $B_2 \subseteq A_2^0$ be the sets of goods in the respective bundles that are exchanged between agents $1$ and $2$, $B_{1,3} \subseteq A_1^0$ and $B_3 \subseteq A_3^0$ be the sets of goods in the respective bundles that are exchanged between agents $1$ and $3$, and let $C_1 = A_1^0 \setminus (B_{1,2} \cup B_{1,3})$, $C_2 = A_2^0 \setminus B_2$, and $C_3 = A_3^0 \setminus B_3$.
Note that all these sets are disjoint by \Cref{lem:goods}.
Let $x = |B_{1,2}| = |B_2|$ and $y = |B_{1,3}| = |B_3|$.
We have $x+y = 2s/3$, $|C_2| = s-x$, $|C_3| = s-y$, and $|C_1| = s-x-y = s/3$.
Without loss of generality, let $x \ge y$.
Note that $x \leq s/2$, since otherwise agent $2$ will be weak-EF1 towards agent $1$ by \Cref{lem:worst_iden_two_agents} and does not need to exchange more goods with agent $1$.

In $\mathcal{A}^0$, we have $A_1^0 = B_{1,2} \cup B_{1,3} \cup C_1$, $A_2^0 = B_2 \cup C_2$, and $A_3^0 = B_3 \cup C_3$.
In $\mathcal{A}^T$, we have $A_1^T = B_2 \cup B_3 \cup C_1$, $A_2^T = B_{1,2} \cup C_2$, and $A_3^T = B_{1,3} \cup C_3$.
Since the algorithm always exchanges the most valuable goods from agent $1$'s bundle and the least valuable goods from agent $2$'s and agent $3$'s bundles, we have 
\begin{align*}
\frac{u(B_{1,2})}{x} &\geq \frac{u(C_1)}{s/3}, \quad \frac{u(B_{1,3})}{y} \geq \frac{u(C_1)}{s/3}, \quad
\frac{u(C_2)}{s-x} \geq \frac{u(B_2)}{x}, \quad\text{ and }\quad \frac{u(C_3)}{s-y} \geq \frac{u(B_3)}{y}.
\end{align*}
By \Cref{lem:goods_comparison}, we have $u(B_{1,2}) \geq u(B_2)$ and $u(B_{1,3}) \geq u(B_3)$.

Since $x\ge y$, we have $s/3 \leq x \leq s/2$ and hence $s/6 \leq y \leq s/3$.
Let 
\begin{align*}
\alpha = \frac{6x-s}{3x+s} = 2 - \frac{3s}{3x+s} = 2 - \frac{s}{s-y}.
\end{align*}
Since $s/3 \leq x \leq s/2$, we have $1/2 \leq \alpha \leq 4/5$.

Let $\alpha_2 = \alpha(s-x)/x$ and $\alpha_3 = (1-\alpha)(s-y)/y$.
Since $1/3 \leq x/s \leq 1/2$, we have $\alpha \geq x/s$, which implies that 
\begin{align*}
\alpha_2 = \frac{\alpha(s-x)}{x} = \frac{\alpha s}{x} - \alpha \geq 1 - \alpha.
\end{align*}
On the other hand, the derivative of $\alpha_2 = \alpha (s-x)/x$ with respect to $x$ is
\begin{align*}
    \left( \frac{6x-s}{3x+s} \right) \left( -\frac{s}{x^2} \right) + \left( \frac{9s}{(3x+s)^2} \right)\left( \frac{s-x}{x} \right) 
    &= \left( \frac{s}{x(3x+s)} \right) \left( \frac{s-6x}{x} + \frac{9s-9x}{3x+s} \right) \\
    &= \left( \frac{s}{x(3x+s)} \right) \left( \frac{(s+9x)(s-3x)}{x(3x+s)} \right).
\end{align*}
When the derivative of $\alpha_2$ with respect to $x$ is equal to $0$, we get $x = -s/9$ or $x = s/3$.
It can be verified that $\alpha_2$ attains a local maximum at $x = s/3$.
For $x \in [s/3, s/2]$, the maximum value of $\alpha_2$ is hence equal to $1$ at $x = s/3$.
Together, we have $1-\alpha \leq \alpha_2 \leq 1$.

Now,
\begin{align*}
    \alpha_3 &= (1-\alpha)\frac{s-y}{y} 
    = \left( 1 - \left( 2 - \frac{s}{s-y} \right) \right) \frac{s-y}{y} 
    = -\frac{s-y}{y} + \frac{s}{y} 
    = 1.
\end{align*}
We shall show that $\alpha u(A_2^T) + (1-\alpha)u(A_3^T) \geq u(A_1^T)$.
We have
\begin{align*}
    \alpha u(A_2^T) &+ (1-\alpha) u(A_3^T) \\
    &= \alpha(u(B_{1,2})+u(C_2)) + (1-\alpha)(u(B_{1,3})+u(C_3)) \\
    &\geq \alpha u(B_{1,2}) + \frac{\alpha(s-x)}{x} u(B_2) + (1-\alpha)u(B_{1,3}) + \frac{(1-\alpha)(s-y)}{y} u(B_3) \\
    &= \alpha u(B_{1,2}) + \alpha_2 u(B_2) + (1-\alpha)u(B_{1,3}) + \alpha_3 u(B_3) \\
    &= (\alpha + \alpha_2 - 1) u(B_{1,2}) + (1-\alpha_2)u(B_{1,2}) + \alpha_2 u(B_2) + (1-\alpha) u(B_{1,3}) + u(B_3) \\
    &\geq (\alpha + \alpha_2 - 1) u(B_{1,2}) + (1-\alpha_2)u(B_2) + \alpha_2 u(B_2) + (1-\alpha) u(B_{1,3}) + u(B_3) \\
    &= (\alpha + \alpha_2 - 1) u(B_{1,2}) + u(B_2)  + (1-\alpha) u(B_{1,3}) + u(B_3) \\
    &\geq (\alpha + \alpha_2 - 1) \frac{3x}{s}u(C_1) + u(B_2)  + (1-\alpha) \frac{3y}{s}u(C_1) + u(B_3) \\
    &= \frac{3}{s} \left[ (\alpha+\alpha_2)x - x + (1-\alpha)y \right]u(C_1) + u(B_2) + u(B_3).
\end{align*}
Since $\alpha_2x = \alpha(s-x)$ implies $(\alpha+\alpha_2)x = \alpha s$ and $y = \alpha_3y = (1-\alpha)(s-y)$ implies $(1-\alpha)y = (1-\alpha)s - y$, the expression $(\alpha+\alpha_2)x - x + (1-\alpha)y$ simplifies to $\alpha s - x + (1-\alpha)s - y$, which gives $s-x-y$.
Using the fact that $x+y = 2s/3$, the expression simplifies to $s/3$.
Therefore,
\begin{align*}
    \alpha u(A_2^T) + (1-\alpha) u(A_3^T) 
    &\geq \frac{3}{s}\left(\frac{s}{3}\right)u(C_1) + u(B_2) + u(B_3) 
    = u(C_1) + u(B_2) + u(B_3) 
    = u(A_1^T).
\end{align*}

Let $j \in \argmax_{k\in\{1,2,3\}} u(A_k^T)$.
Since $\alpha u(A_2^T) + (1-\alpha) u(A_3^T) \geq u(A_1^T)$ for some $\alpha \in (0, 1)$, we may assume that $j \in J^T$.
Suppose without loss of generality that $j = 2$.
Note that agent $2$ is weak-EF1 towards every other agent in $\mathcal{A}^T$.
Agent $1$ is weak-EF1 towards every other agent in $\mathcal{A}^T$ by \Cref{lem:agents_ef1}.
Let $t < T$ be the round that agent $2$ exchanges a good with agent $1$ for the final time, i.e., agent~$2$ exchanges a good with agent $1$ going from $\mathcal{A}^t$ to $\mathcal{A}^{t+1}$.
Then, by \Cref{lem:agents_ef1}, agent $3$ is weak-EF1 towards agent $2$ in $\mathcal{A}^{t+1}$.
Since the utility of agent $3$'s bundle does not decrease thereafter and agent~$2$'s bundle remains the same thereafter, agent $3$ is weak-EF1 towards agent $2$ in $\mathcal{A}^T$.
Then, agent $3$ is weak-EF1 towards every other agent in $\mathcal{A}^T$.
This shows that $\mathcal{A}^T$ is weak-EF1, contradicting the original assumption.

\textbf{Case 2: $|I^T| = 2$. }
We consider the allocation $\mathcal{A}^T$ relative to $\mathcal{A}^0$.
Without loss of generality, let $1 \in J^T$.
Let $B_{1,2} \subseteq A_1^0$ and $B_2 \subseteq A_2^0$ be the sets of goods in the respective bundles that are exchanged between agents $1$ and $2$, $B_{1,3} \subseteq A_1^0$ and $B_3 \subseteq A_3^0$ be the sets of goods in the respective bundles that are exchanged between agents $1$ and $3$, and let $C_1 = A_1^0 \setminus (B_{1,2} \cup B_{1,3})$, $C_2 = A_2^0 \setminus B_2$, and $C_3 = A_3^0 \setminus B_3$.
Note that all these sets are disjoint by \Cref{lem:goods}.
Let $x = |B_{1,2}| = |B_2|$ and $y = |B_{1,3}| = |B_3|$.
We have $x+y = 2s/3$, $|C_2| = s-x$, $|C_3| = s-y$, and $|C_1| = s-x-y = s/3$.
Without loss of generality, let $x \geq y$.
Note that $x \leq s/2$, since otherwise agent $1$ will be weak-EF1 towards agent $2$ by \Cref{lem:worst_iden_two_agents} and does not need to exchange more goods with agent $2$.

In $\mathcal{A}^0$, we have $A_1^0 = B_{1,2} \cup B_{1,3} \cup C_1$, $A_2^0 = B_2 \cup C_2$, and $A_3^0 = B_3 \cup C_3$.
In $\mathcal{A}^T$, we have $A_1^T = B_2 \cup B_3 \cup C_1$, $A_2^T = B_{1,2} \cup C_2$, and $A_3^T = B_{1,3} \cup C_3$.
Since the algorithm always exchanges the least valuable goods from agent $1$'s bundle and the most valuable goods from agent $2$'s and agent $3$'s bundles, we have 
\begin{align*}
\frac{u(B_{1,2})}{x} \leq \frac{u(C_1)}{s/3}, \quad \frac{u(B_{1,3})}{y} \leq \frac{u(C_1)}{s/3}, \quad \frac{u(C_2)}{s-x} \leq \frac{u(B_2)}{x}, \quad \text{ and } \quad \frac{u(C_3)}{s-y} \leq \frac{u(B_3)}{y}.
\end{align*}
By \Cref{lem:goods_comparison}, we have $u(B_{1,2}) \leq u(B_2)$ and $u(B_{1,3}) \leq u(B_3)$.

Since $x \geq y$, we have $s/3 \leq x \leq s/2$.
Let 
\begin{align*}
\alpha = \frac{6x-s}{3x+s} = 2 - \frac{3s}{3x+s} = 2 - \frac{s}{s-y},
\end{align*}
$\alpha_2 = \alpha(s-x)/x$, and $\alpha_3 = (1-\alpha)(s-y)/y$.
By the same reasoning as in Case 1, we have $1/2 \leq \alpha \leq 4/5$, $1-\alpha \leq \alpha_2 \leq 1$, and $\alpha_3 = 1$.

We shall show that $\alpha u(A_2^T) + (1-\alpha)u(A_3^T) \leq u(A_1^T)$.
We have
\begin{align*}
    & \alpha u(A_2^T) + (1-\alpha) u(A_3^T) \\
    &= \alpha(u(B_{1,2})+u(C_2)) + (1-\alpha)(u(B_{1,3})+u(C_3)) \\
    &\leq \alpha u(B_{1,2}) + \frac{\alpha(s-x)}{x} u(B_2) + (1-\alpha)u(B_{1,3}) + \frac{(1-\alpha)(s-y)}{y} u(B_3) \\
    &= \alpha u(B_{1,2}) + \alpha_2 u(B_2) + (1-\alpha)u(B_{1,3}) + \alpha_3 u(B_3) \\
    &= (\alpha + \alpha_2 - 1) u(B_{1,2}) + (1-\alpha_2)u(B_{1,2}) + \alpha_2 u(B_2) + (1-\alpha) u(B_{1,3}) + u(B_3) \\
    &\leq (\alpha + \alpha_2 - 1) u(B_{1,2}) + (1-\alpha_2)u(B_2) + \alpha_2 u(B_2) + (1-\alpha) u(B_{1,3}) + u(B_3) \\
    &= (\alpha + \alpha_2 - 1) u(B_{1,2}) + u(B_2) + (1-\alpha) u(B_{1,3}) + u(B_3) \\
    &\leq (\alpha + \alpha_2 - 1) \frac{3x}{s}u(C_1) + u(B_2) + (1-\alpha) \frac{3y}{s}u(C_1) + u(B_3) \\
    &= \frac{3}{s} \left[ (\alpha+\alpha_2)x - x + (1-\alpha)y \right]u(C_1) + u(B_2) + u(B_3).
\end{align*}
By the same reasoning as in Case 1, we have $(\alpha+\alpha_2)x - x + (1-\alpha)y = s/3$, and therefore, $\alpha u(A_2^T) + (1-\alpha)u(A_3^T) \leq u(A_1^T)$.

Let $i \in \argmin_{k\in\{1,2,3\}} u(A_k^T)$.
Since $\alpha u(A_2^T) + (1-\alpha) u(A_3^T) \leq u(A_1^T)$ for some $\alpha \in (0, 1)$, we may assume that $i \in I^T$.
Suppose without loss of generality that $i = 2$.
Note that every agent is weak-EF1 towards agent~$2$ in $\mathcal{A}^T$.
Every agent is weak-EF1 towards agent $1$ in $\mathcal{A}^T$ by \Cref{lem:agents_ef1}.
Let $t < T$ be the round that agent $2$ exchanges a good with agent $1$ for the final time, i.e., agent~$2$ exchanges a good with agent $1$ going from $\mathcal{A}^t$ to $\mathcal{A}^{t+1}$.
Then, by \Cref{lem:agents_ef1}, agent $2$ is weak-EF1 towards agent~$3$ in $\mathcal{A}^{t+1}$.
Since the utility of agent $3$'s bundle does not increase thereafter and agent $2$'s bundle remains the same thereafter, agent~$2$ is weak-EF1 towards agent $3$ in $\mathcal{A}^T$.
Then, every agent is weak-EF1 towards agent $3$ in $\mathcal{A}^T$.
This shows that $\mathcal{A}^T$ is weak-EF1, contradicting the original assumption.
\end{proof}

We are now ready to show the result on $\fidweak(n, s)$ for three agents.

\begin{theorem}
Let $s$ be a positive integer divisible by $3$.
Then, $\fidweak(3, s) = 2s/3$.
\end{theorem}

\begin{proof}
The lower bound of $\fidweak(3, s)$ follows from \Cref{thm:worst_idenbin}---note that weak-EF1 and EF1 are equivalent for binary utilities.
The upper bound follows from \Cref{lem:worst_iden_three_agents}.
\end{proof}

\section{Beneficial Exchanges}
\label{ap:beneficial}

Let us say that an exchange is \emph{beneficial} if the two agents involved in the exchange strictly benefit from the exchange, i.e., if the goods $g \in A_i$ and $g' \in A_j$ are exchanged, then $u_i(g') > u_i(g)$ and $u_j(g) > u_j(g')$.
In this appendix, we investigate the decision problem of whether a given initial allocation can be reformed into an EF1 allocation using \emph{only} beneficial exchanges.
For convenience, we refer to this problem as \textsc{Beneficial Exchanges}.

We show that \textsc{Beneficial Exchanges} is NP-complete, even for binary utilities, using a reduction from \textsc{Minimum $k$-Coverage}.
In \textsc{Minimum $k$-Coverage}, we are given positive integers $k, \ell, p, q$ such that $k \leq q$ and $\ell \leq p$, a set $X = \{x_1, \ldots, x_q\}$, and a collection $C = \{Y_1, \ldots, Y_p\}$ of subsets of $X$.
The problem is to decide whether there exists a set $I \subseteq \{1, \ldots, p\}$ of indices such that $|I| = \ell$ and $|\bigcup_{i \in I} Y_i| \leq k$.
This decision problem is known to be NP-hard \citep{Vinterbo02}.

\begin{theorem}
\label{thm:beneficial}
\textsc{Beneficial Exchanges} is NP-complete for binary utilities.
\end{theorem}

\begin{proof}
For membership in NP, observe that in a sequence of beneficial exchanges for binary utilities, each good $g \in M$ can only be part of at most one exchange.
Indeed, if good $g$ is part of at least two beneficial exchanges, then it must be received by some agent $i$ (and hence worth $1$ to $i$) and be given away by agent $i$ (and hence worth $0$ to $i$), which is impossible.
Therefore, such a sequence consists of at most $m/2$ exchanges, and can be used as a certificate for polynomial-time verification.

It remains to show that the problem is NP-hard.
Let an instance of \textsc{Minimum $k$-Coverage} be given.
Define an instance of \textsc{Beneficial Exchanges} as follows.
There are $n = 2p+q+k-\ell$ agents and $m = 2n$ goods.
We shall label the agents $a_{1,1}, \ldots, a_{1,q}$, $a_{2,1}, \ldots, a_{2,k}$, $a_{3,1}, \ldots, a_{3,p}$, $a_{4,1}, \ldots, a_{4,p-\ell}$; we use $u_{i,j}$ for the utility of agent $a_{i,j}$.
For each agent $a_{i,j}$, there are two goods $g_{i,j}^0$ and~$g_{i,j}^1$ that are both in agent $a_{i,j}$'s bundle in the initial allocation.
The valuable goods for the agents are as follows:
\begin{itemize}
    \item For $i \in \{1, \ldots, q\}$, $u_{1,i}(g_{2,j}^1) = 1$ for all $j \in \{1, \ldots, k\}$.
    Additionally, if $x_i \in Y_j$ for some $j \in \{1, \ldots, p\}$, then $u_{1,i}(g_{3,j}^0) = u_{1,i}(g_{3,j}^1) = 1$.
    \item For $i \in \{1, \ldots, k\}$, $u_{2,i}(g_{1,j}^1) = 1$ for all $j \in \{1, \ldots, q\}$.
    \item For $i \in \{1, \ldots, p\}$, $u_{3,i}(g_{4,j}^1) = 1$ for all $j \in \{1, \ldots, p-\ell\}$.
    \item For $i \in \{1, \ldots, p-\ell\}$, $u_{4,i}(g_{3,j}^1) = 1$ for all $j \in \{1, \ldots, p\}$.
\end{itemize}
All other goods not mentioned above are worth $0$ to the respective agents.
This reduction can be done in polynomial time.

In the initial allocation, every agent has zero utility for her own bundle, and the only agents who are possibly not EF1 are agents $a_{1,i}$, who envy $a_{3,j}$ if $x_i \in Y_j$.
By construction, the only possible beneficial exchanges are between $g_{1,i}^1$ in agent $a_{1,i}$'s bundle and $g_{2,j}^1$ in agent $a_{2,j}$'s bundle, or between $g_{3,i}^1$ in agent~$a_{3,i}$'s bundle and $g_{4,j}^1$ in agent $a_{4,j}$'s bundle.

We claim that the initial allocation can be reformed into an EF1 allocation via only beneficial exchanges if and only if there exists a set $I \subseteq \{1, \ldots, p\}$ of indices such that $|I| = \ell$ and $|\bigcup_{i \in I} Y_i| \leq k$.

$(\Leftarrow)$ Suppose that there exists a set $I \subseteq \{1, \ldots, p\}$ of indices such that $|I| = \ell$ and $|\bigcup_{i \in I} Y_i| \leq k$.
\begin{itemize}
    \item Let $I' = \{1, \ldots, p\} \setminus I$.
    Since $|I'| = p-\ell$, there exists a bijection $\sigma : I' \to \{1, \ldots, p-\ell\}$.
    For each $i' \in I'$, exchange $g_{3,i'}^1$ in agent $a_{3,i'}$'s bundle with $g_{4,\sigma(i')}^1$ in agent $a_{4,\sigma(i')}$'s bundle.
    \item Let $J = \{j \mid x_j \in \bigcup_{i \in I} Y_i \}$.
    Since $|J| \leq k$, there exists an injection $\phi : J \to \{1, \ldots, k\}$.
    For each $j \in J$, exchange $g_{1,j}^1$ in agent $a_{1,j}$'s bundle with $g_{2,\phi(j)}^1$ in agent $a_{2,\phi(j)}$'s bundle.
\end{itemize}
We now show that the new allocation is EF1.
It is easy to see that the allocation is EF1 for agents $a_{2,i}$, $a_{3,i}$, and $a_{4,i}$, since every other agent has at most one of their valuable goods.
Therefore, it suffices to show that the allocation is EF1 for agents $a_{1,j}$.
\begin{itemize}
    \item If $j \in J$, then agent $a_{1,j}$ has the valuable good $g_{2,\phi(j)}^1$ in her bundle, so her utility of her own bundle is at least $1$.
    Since her utility of every other agent's bundle is at most $2$, agent $a_{1,j}$ is EF1 towards every other agent.
    \item If $j \notin J$, then $x_j \notin \bigcup_{i \in I} Y_i$, and so $x_j \notin Y_i$ for all $i \in I$.
    Note that the valuable goods for $a_{1,j}$ are possibly in the form $g_{2,i}^1$, $g_{3,i}^0$, and $g_{3,i}^1$.
    Suppose on the contrary that $a_{1,j}$ is not EF1 towards some agent.
    This agent must have two such goods in the final allocation.
    The only way for this to happen is when there exists $i^*$ such that agent $a_{3,i^*}$ has both $g_{3,i^*}^0$ and $g_{3,i^*}^1$.
    This means agent~$a_{3,i^*}$ had not exchanged any goods, and so $i^* \notin I'$.
    This implies that $i^* \in I$.
    Since $x_j \notin Y_i$ for all $i \in I$, we must have $u_{1,j}(g_{3,i^*}^0) = u_{1,j}(g_{3,i^*}^1) = 0$.
    This contradicts the assumption that $a_{1,j}$ is not EF1 towards agent $a_{3,i^*}$.
    Therefore, $a_{1,j}$ is EF1 towards every agent.
\end{itemize}

$(\Rightarrow)$ Suppose that the initial allocation can be reformed into an EF1 allocation via only beneficial exchanges.
Consider one such sequence of beneficial exchanges.
Let $I' \subseteq \{1, \ldots, p\}$ be the set of all indices $i'$ such that agent $a_{3,i'}$ exchanged a good with another agent in this sequence.
Since $a_{3,i'}$ can only exchange a good with some $a_{4,i''}$ once, and there are only $p-\ell$ agents of the form $a_{4,i''}$, we have $|I'| \leq p-\ell$.
Therefore, $I_0 := \{1, \ldots, p\} \setminus I'$ has cardinality at least $\ell$, and $I_0$ contains indices $i$ such that agent $a_{3,i}$ retains her original bundle from the initial allocation.

We claim that $|\bigcup_{i \in I_0} Y_i| \leq k$.
Let $J \subseteq \{1, \ldots, q\}$ be the set of all indices $j$ such that agent~$a_{1,j}$ exchanged a good with another agent in this sequence.
Since $a_{1,j}$ can only exchange a good with some $a_{2,j''}$ once, and there are only $k$ agents of the form $a_{2,j''}$, we have $|J| \leq k$.
Therefore, $J' := \{1, \ldots, q\} \setminus J$ has cardinality at least $q-k$, and $J'$ contains all indices $j'$ such that agent $a_{1,j'}$ retains her original bundle from the initial allocation; these agents have utility $0$.
Since the final allocation is EF1, these agents do not envy agents $a_{3,i}$ by more than one good for each $i \in I_0$.
Therefore, we must have $u_{1,j'}(g_{3,i}^0) = u_{1,j'}(g_{3,i}^1) = 0$, which implies that $x_{j'} \notin Y_i$ for all $j'\in J'$ and $i\in I_0$.
This means that $x_{j'} \notin \bigcup_{i \in I_0} Y_i$ for every $j'\in J'$.
Therefore, at least $q-k$ of the $x_{j}$'s are not in $\bigcup_{i \in I_0} Y_i$, which shows that $\bigcup_{i \in I_0} Y_i$ has cardinality at most $q-(q-k) = k$, as claimed.

Finally, take any subset $I \subseteq I_0$ with cardinality $\ell$.
The proof is completed by noting that $\bigcup_{i \in I} Y_i \subseteq \bigcup_{i \in I_0} Y_i$.
\end{proof}

While we have shown that a sequence of beneficial exchanges must be of polynomial length for binary utilities, the same statement in fact holds for general utilities.
Indeed, for any sequence of beneficial exchanges, for each good $g_{t_1}$ in some agent $i$'s initial bundle, it is exchanged with another good~$g_{t_2}$, which is subsequently exchanged with another good $g_{t_3}$, and so on, until some $g_{t_k}$ in agent $i$'s final bundle.
Since we must have $u_i(g_{t_1}) < \dots < u_i(g_{t_k})$ due to the exchanges being beneficial, it must hold that $k \leq m$, and so there are at most $m-1$ exchanges starting from $g_{t_1}$.
Since there are $m$ goods in total and each exchange involves two goods, the maximum number of exchanges in the sequence is $m(m-1) /2$.
Hence, by \Cref{thm:beneficial}, we have NP-completeness for general utilities as well.

\section{Transfers}
\label{ap:transfer}

Thus far, we have focused on the operation where two agents can exchange a pair of goods with each other.
In this appendix, we consider a setting where we instead allow an agent to \emph{transfer} a good to another agent.\footnote{An interesting extension would be to allow both exchanges and transfers (see, e.g., \citep[App.~A]{IgarashiKaSu24}), but we do not consider it here.}
Note that with transfers, it is possible to reach any allocation from any other allocation---indeed, with each transfer, we can move any good that is still not with the ``correct'' agent to that agent.
Hence, the corresponding \textsc{Reformability} problem is trivial, as the answer is always simply ``Yes''.
We shall therefore concentrate on the complexity of computing the optimal number of transfers required to reach an EF1 allocation, as well as worst-case bounds.
As we will show, we can obtain several results in a similar manner as for exchanges (\Cref{sec:optimal}).

For convenience, we refer to as \textsc{Optimal Transfers} the problem of deciding---given an instance, an initial allocation in the instance, and a number $k$---whether the optimal number of transfers required to reach an EF1 allocation is at most $k$.
From the previous paragraph, this number is always polynomial in the number of agents and the number of goods.
Therefore, \textsc{Optimal Transfers} is in NP.

We begin by showing that, for two agents with identical utilities, the problem is also in P.

\begin{theorem}
\textsc{Optimal Transfers} is in P for two agents with identical utilities.    
\end{theorem}

\begin{proof}
We show that we can compute the optimal number of transfers 
in polynomial time. 
If the initial allocation ${\cal A}$ is EF1, we are done. 
Otherwise, assume without loss of generality that agent~$2$
has a higher utility than agent~$1$ in ${\cal A}$. 
The algorithm proceeds as follows: repeatedly transfer 
a most valuable good in agent~$2$'s bundle 
to agent~$1$ until agent~$1$ is EF1 towards agent~$2$. 
The optimal number of transfers required is then the
number of transfers made in this algorithm.

If agent~$2$'s bundle is empty, then 
agent~$1$ is EF1 towards agent~$2$.
Thus, the total number of transfers in the algorithm 
is at most~$m$. 
Each transfer can be performed in polynomial time, 
and so the algorithm terminates in polynomial time.
We show next that an EF1 allocation is obtained when the
algorithm terminates.
To this end, it is sufficient to prove that agent~$2$ is EF1 towards
agent~$1$ at every step of the algorithm. 
Let the initial allocation be ${\cal A}^0 = {\cal A}$, 
and let ${\cal A}^t$ be the allocation after $t$ steps of
the algorithm. 
Note that $\mathcal{A}^0$ satisfies the condition that agent $2$ is EF1 towards agent $1$, since agent $2$ has a higher utility than agent $1$ in $\mathcal{A}$.
We show that if
${\cal A}^t$ has the property that agent~$2$ is 
EF1 towards agent~$1$ and agent~$1$ is \emph{not} EF1 towards agent~$2$, then ${\cal A}^{t+1}$ has the property 
that agent~$2$ is EF1 towards agent~$1$. 
Suppose that $g \in A^t_2$
is transferred to agent~$1$. 
Since agent~$1$ is not EF1 towards agent~$2$ in ${\cal A}^t$, it holds that 
$u(A^t_1) < u(A^t_2 \setminus \{g\})$. 
Thus, 
\begin{equation*}
u(A^{t+1}_2) 
= u(A^t_2 \setminus \{g\})
> u(A^t_1) = u(A^{t+1}_1 \setminus \{g\}),
\end{equation*} 
showing that agent~$2$ is EF1 towards agent~$1$ in ${\cal A}^{t+1}$. 

Finally, we show that the optimal number of 
transfers required to reach an EF1 allocation
is at least the number of transfers made in this algorithm. 
Let $T$ be the number of
transfers made in this algorithm.
For each $t \in \{1,\dots,T\}$, 
let $g^t \in A_2$ be the good in agent~$2$'s 
bundle that is transferred at the $t^\text{th}$ step
of the algorithm.
Notice that 
$u(A^{T-1}_1) < u(A^{T-1}_2 \setminus \{g^T\})$, where $g^T$ is a good with the highest utility in agent $2$'s bundle in $\mathcal{A}^{T-1}$. 
Suppose on the contrary that only $k \le T-1$
transfers  are required to reach an
EF1 allocation. Since ${\cal A}$ is not EF1, 
we have $1 \le k < T$.
Let $(B_1, B_2)$ 
be the EF1 allocation
after the $k$ transfers. 
The utility of $B_1$ is upper-bounded by the utility of $A_1$
after adding $k$ goods of the highest utility from $A_2$, so 
\begin{equation*}
u(B_1) 
\le u(A_1 \cup \{g^1,\dots,g^k\})
\le u(A_1 \cup \{g^1,\dots,g^{T-1}\}) 
= u(A^{T-1}_1). 
\end{equation*}
On the other hand, the utility of 
$B_2$ without the most valuable good is lower-bounded
by the utility of~$A_2$ after 
removing $k + 1 \le T$ goods of the highest utility from 
$A_2$, so 
\begin{equation*}
u(B_2 \setminus \{g\}) 
\ge 
u(A_2 \setminus \{g^1,\dots,g^{k+1}\})
\ge
u(A_2 \setminus \{g^1,\dots,g^{T}\})
= 
u(A_2^{T-1} \setminus \{g^{T}\})
\end{equation*} 
for every $g \in B_2$. 
This implies that
\begin{equation*}
u(B_1) \le u(A^{T-1}_1) < u(A^{T-1}_2 
\setminus \{g^T\}) \le u(B_2 \setminus \{g\})
\end{equation*}
for all $g \in B_2$. Hence, 
agent~$1$ is not EF1 towards agent~$2$ in $(B_1, B_2)$, 
contradicting the
assumption that $(B_1, B_2)$ is 
EF1. It follows that at least $T$ transfers 
are required to reach 
an EF1 allocation.   
\end{proof}

If the two agents do not have identical utilities, we show that the problem becomes computationally intractable.

\begin{theorem}
\label{thm:optimal-transfer-two-NP}
\textsc{Optimal Transfers} is NP-complete for two agents.     
\end{theorem}

\begin{proof}
Clearly, this problem is in NP.
To demonstrate NP-hardness, we modify the construction from the proof of \Cref{thm:optimal_gen_two_nphard}.
Recall that we have an instance $\mathcal{I}$ with $n = 2$ agents and a set of goods $M = \{g_1, \ldots, g_{4q+12}\}$, where $u_i(g_j) = 0$ for all $i \in \{1, 2\}$ and $j \in \{2q+7, \ldots, 4q+12\}$.
In the initial allocation $\mathcal{A}$, agent $1$ has $A_1 = \{g_{2q+7}, \ldots, g_{4q+12}\}$ and agent $2$ has $A_2 = \{g_1, \ldots, g_{2q+6}\}$.
In \Cref{thm:optimal_gen_two_nphard}, it was proven that it is NP-hard to decide whether the optimal number of exchanges required to reach an EF1 allocation from $\mathcal{A}$ is at most $q+2$ in $\mathcal{I}$.

Define an instance $\mathcal{I}'$ of \textsc{Optimal Transfers} as follows.
There are $n = 2$ agents and a set of goods $M = \{g_1, \ldots, g_{2q+6}\}$.
For $j \in \{1, \ldots, 2q+6\}$, the utility of $g_j$ for each agent is identical to that in the original instance $\mathcal{I}$.
In the initial allocation~$\mathcal{A}'$, agent~$2$ has all of the $2q+6$ goods.
This reduction can be done in polynomial time.
We claim that the optimal number of transfers required to reach an EF1 allocation from~$\mathcal{A}'$ in $\mathcal{I}'$ is the same as the optimal number of exchanges required to reach an EF1 allocation from $\mathcal{A}$ in $\mathcal{I}$.
Indeed, if $k$~exchanges are sufficient to reach an EF1 allocation from~$\mathcal{A}$ in $\mathcal{I}$ for some~$k$, then transferring the $k$~corresponding goods from agent~$2$ to agent~$1$ in $\mathcal{I}'$ leads to an EF1 allocation.
Likewise, if $k$~transfers are sufficient to reach an EF1 allocation from $\mathcal{A}'$ in $\mathcal{I}'$, then exchanging the $k$~corresponding goods with $k$~arbitrary goods in agent~$1$'s bundle in $\mathcal{I}$ leads to an EF1 allocation.
Hence, the NP-hardness shown in \Cref{thm:optimal_gen_two_nphard} (when $k = q+2$) implies that \textsc{Optimal Transfers} is also NP-hard.    
\end{proof}

For a constant number of agents with identical utilities, the NP-hardness construction in the proof of \Cref{thm:optimal_iden_const} also has the structure that, in the initial allocation, all agents except one only have goods of utility~$0$.
Therefore, an analogous reduction as in the proof of \Cref{thm:optimal-transfer-two-NP} implies that the problem is NP-hard.

\begin{theorem}
\textsc{Optimal Transfers} is NP-complete for $n\ge 3$ agents with identical utilities, where $n$ is a constant.     
\end{theorem}

On the other hand, for a constant number of agents with binary utilities, the problem can be solved efficiently.
Note that unlike for exchanges, where it is NP-hard in general to determine the smallest number of exchanges required to reach one allocation from another allocation (without EF1 considerations) \citep{IgarashiKaSu24}, for transfers this task is trivial, as the required number is simply equal to the number of goods that belong to different agents in the two given allocations.

\begin{theorem}
\textsc{Optimal Transfers} is in P for a constant number of agents with binary utilities.
\end{theorem}

\begin{proof}
Recall that for binary utilities, there are $2^n$ types of goods.
For $i\in\{1,\dots,n\}$ and $j\in\{1,\dots,2^n\}$, let $a_{i,j}$ denote the number of goods of type~$j$ in agent~$i$'s bundle.
Using the same approach as in \Cref{lem:enumerate_bin_const}, we can enumerate all equivalence classes of EF1 allocations (even without a specified size vector) in polynomial time.
Consider one such equivalence class; for $i\in\{1,\dots,n\}$ and $j\in\{1,\dots,2^n\}$, let $b_{i,j}$ denote the number of goods of type~$j$ in agent~$i$'s bundle for an allocation in this equivalence class.
The optimal number of transfers required to reach some allocation from this equivalence class is then $\frac{1}{2}\sum_{i=1}^n\sum_{j=1}^{2^n}|a_{i,j}-b_{i,j}|$.
The smallest such number across all considered equivalence classes will then answer the decision problem of \textsc{Optimal Transfers}.
\end{proof}

We now consider a general number of agents.
For the instance constructed in the proof of \Cref{thm:optimal_bin_gen}, a similar argument shows that the optimal number of \emph{transfers} required to reach an EF1 allocation from~$\mathcal{A}$ is at most~$q$ if and only if there exists an exact cover in~$C$. (In particular, for each $Y_j\in D$, we select some $x_i\in Y_j$ and transfer the good~$h_j$ from agent~$3q+1$ to agent~$i$.)

\begin{theorem}
\textsc{Optimal Transfers} is NP-complete for binary utilities.
\end{theorem}

For identical binary utilities, we can ignore all non-valuable goods and, as long as the allocation is not EF1, transfer a good from an agent with the largest number of goods to an agent with the smallest number of goods.

\begin{theorem}
\textsc{Optimal Transfers} is in P for identical binary utilities.
\end{theorem}

Finally, we consider worst-case bounds.
Given $n$ and~$s$, let $\ftrans(n,s)$ be the smallest integer such that for every instance with $n$ agents and $ns$ goods and every $s$-balanced allocation $\mathcal{A}$ in the instance, there exists an EF1 allocation that can be reached from~$\mathcal{A}$ using at most $\ftrans(n, s)$ transfers.
We focus on the two cases in which our bound for $f(n,s)$ in \Cref{thm:worst_ub} is tight (according to \Cref{thm:worst_lb_bin}), and show that in these cases, $\ftrans(n, s)$ is exactly twice of $f(n, s)$.

\begin{theorem}
Let $n$ and $s$ be positive integers.
\begin{enumerate}
    \item[(a)] If $s$ is divisible by~$n$, then $\ftrans(n,s) = s(n-1)$.
    \item[(b)] If $s$ is not divisible by~$2$, then $\ftrans(2,s) = s-1$.
\end{enumerate}
\end{theorem}

\begin{proof}
\begin{enumerate}
\item[(a)] Suppose that $s$ is divisible by~$n$.
For the lower bound, consider an instance where each agent has utility~$0$ for each good in her bundle and utility~$1$ for each good in any other agent's bundle.
Since each agent values $s(n-1)$ goods, she needs to receive a bundle worth at least $s(n-1)/n$ to her in order to be EF1.
Thus, at least $s(n-1)/n$ goods must be transferred \emph{to} her.
Since this is true for each of the $n$~agents, the number of transfers required to reach an EF1 allocation is at least $s(n-1)$.

For the upper bound, consider any instance, and denote the initial allocation by~$\mathcal{A}$.
Using the algorithm of \citet{BiswasBa18}, we can find an $s$-balanced EF1 allocation~$\mathcal{B}$ such that $|A_i\cap B_j| = |A_i|/n = s/n$ for all $i,j\in N$.
In order to reach $\mathcal{B}$ from $\mathcal{A}$, for each (ordered) pair of distinct agents $i,j\in N$, we need to transfer $s/n$ goods from $i$ to~$j$.
Hence, the number of transfers required to reach $\mathcal{B}$ from~$\mathcal{A}$ is $n(n-1)\cdot s/n = s(n-1)$.

\item[(b)] Suppose that $s$ is not divisible by~$2$.
For the lower bound, consider an instance where each agent has utility~$0$ for each good in her bundle and utility~$1$ for each good in the other agent's bundle.
Since each agent values $s$ goods, she needs to receive a bundle worth at least $(s-1)/2$ to her in order to be EF1.
Thus, at least $(s-1)/2$ goods must be transferred \emph{to} her.
Since this is true for each of the two agents, the number of transfers required to reach an EF1 allocation is at least $s-1$.

For the upper bound, consider any instance, and denote the initial allocation by~$\mathcal{A}$.
Allocate the goods in $A_1$ to the two agents in a round-robin fashion with agent $1$ going first, and allocate the goods in $A_2$ to the two agents in a round-robin fashion with agent $2$ going first.
Call this new allocation $\mathcal{B}$; note that $\mathcal{B}$ is $s$-balanced.
To see that $\mathcal{B}$ is EF1, observe that agent $1$ does not envy agent $2$ with respect to the goods chosen from $A_1$ and is EF1 towards agent $2$ with respect to the goods chosen from $A_2$, so agent $1$ is EF1 towards agent~$2$ in $\mathcal{B}$; likewise, agent~$2$ is EF1 towards agent~$1$ in $\mathcal{B}$.
Since $A_i \cap B_{3-i} = (s-1)/2$ for $i \in \{1, 2\}$, the optimal number of transfers required to reach $\mathcal{B}$ from $\mathcal{A}$ is $2\cdot (s-1)/2 = s-1$.
\qedhere
\end{enumerate}
\end{proof}

\end{document}